\documentclass[sigconf, nonacm]{acmart}

\usepackage{enumitem,kantlipsum}
\usepackage{amsmath}
\usepackage{multirow}
\usepackage{multicol}
\usepackage{graphicx}
\usepackage{subfigure}
\usepackage{caption}
\newtheorem{property}{Property}

\theoremstyle{definition}
\newtheorem{definition}{Definition}[section]
\theoremstyle{plain}
\newtheorem{example}{Example}[section]
\theoremstyle{plain}
\newtheorem{theorem}{Theorem}[section]
\numberwithin{definition}{section}
\numberwithin{example}{section}
\numberwithin{theorem}{section}

\newcommand{\eat}[1]{}

\usepackage[linesnumbered,ruled]{algorithm2e}

\newcommand\vldbdoi{XX.XX/XXX.XX}
\newcommand\vldbpages{XXX-XXX}
\newcommand\vldbvolume{14}
\newcommand\vldbissue{1}
\newcommand\vldbyear{2020}
\newcommand\vldbauthors{\authors}
\newcommand\vldbtitle{\shorttitle} 
\newcommand\vldbavailabilityurl{URL_TO_YOUR_ARTIFACTS}
\newcommand\vldbpagestyle{plain} 
\newcommand{\blue}[1]{\textcolor{blue}{#1}}
\begin{document}

\title{Co-movement Pattern Mining from Videos}


\author{Dongxiang Zhang}
\affiliation{%
  \institution{Zhejiang University, China}
  }
\email{zhangdongxiang@zju.edu.cn}

\author{Teng Ma}
\affiliation{%
  \institution{Zhejiang University, China}
}
\email{mt0228@zju.edu.cn}

\author{Junnan Hu}
\affiliation{%
  \institution{Zhejiang University, China}
}
\email{jnhu@zju.edu.cn}

\author{Yijun Bei}
\affiliation{%
  \institution{Zhejiang University, China}
}
\email{beiyj@zju.edu.cn}

\author{Kian-Lee Tan}
\affiliation{%
  \institution{National University of Singapore}
}
\email{tankl@comp.nus.edu.sg}

\author{Gang Chen}
\affiliation{%
  \institution{Zhejiang University, China}
}
\email{cg@zju.edu.cn}
\renewcommand{\shortauthors}{Trovato and Tobin, et al.}

\begin{abstract}
Co-movement pattern mining from GPS trajectories has been an intriguing subject in spatial-temporal data mining.  In this paper, we extend this research line by migrating the data source from GPS sensors to surveillance cameras, and presenting the first investigation into co-movement pattern mining from videos. We formulate the new problem, re-define the spatial-temporal proximity constraints from cameras deployed in a road network, and theoretically prove its hardness. Due to the lack of readily applicable solutions, we adapt existing techniques and propose two competitive baselines using Apriori-based enumerator and CMC algorithm, respectively. 

As the principal technical contributions, we introduce a novel index called temporal-cluster suffix tree (TCS-tree), which performs two-level temporal clustering within each camera and constructs a suffix tree from the resulting clusters. Moreover, we present a sequence-ahead pruning framework based on TCS-tree,  which allows for the simultaneous leverage of all pattern constraints to filter candidate paths. Finally, to reduce verification cost on the candidate paths, we propose a sliding-window based co-movement pattern enumeration strategy and a hashing-based dominance eliminator, both of which are effective in avoiding redundant operations.

We conduct extensive experiments for scalability and effectiveness analysis. Our results validate the efficiency of the proposed index and mining algorithm, which runs remarkably faster than the two baseline methods. Additionally, we construct a video database with 1169 cameras and perform an end-to-end pipeline analysis to study the performance gap between GPS-driven and video-driven methods. Our results demonstrate that the derived patterns from the video-driven approach are similar to those derived from groundtruth trajectories, providing evidence of its effectiveness.

\end{abstract}

\maketitle

\vspace{-2mm}

\pagestyle{\vldbpagestyle}
\begingroup\small\noindent\raggedright\textbf{PVLDB Reference Format:}\\
\vldbauthors. \vldbtitle. PVLDB, \vldbvolume(\vldbissue): \vldbpages, \vldbyear.\\
\href{https://doi.org/\vldbdoi}{doi:\vldbdoi}
\endgroup
\begingroup
\renewcommand\thefootnote{}\footnote{\noindent
This work is licensed under the Creative Commons BY-NC-ND 4.0 International License. Visit \url{https://creativecommons.org/licenses/by-nc-nd/4.0/} to view a copy of this license. For any use beyond those covered by this license, obtain permission by emailing \href{mailto:info@vldb.org}{info@vldb.org}. Copyright is held by the owner/author(s). Publication rights licensed to the VLDB Endowment. \\
\raggedright Proceedings of the VLDB Endowment, Vol. \vldbvolume, No. \vldbissue\ %
ISSN 2150-8097. \\
\href{https://doi.org/\vldbdoi}{doi:\vldbdoi} \\
}\addtocounter{footnote}{-1}\endgroup

\vspace{-2mm}

\ifdefempty{\vldbavailabilityurl}{}{
\vspace{.3cm}
\begingroup\small\noindent\raggedright\textbf{PVLDB Artifact Availability:}\\
The source code, data, and/or other artifacts have been made available at \url{https://github.com/Mateng0228/Co-movement-Pattern-Mining-from-Videos}.
\endgroup
}

\section{Introduction}
 As a prevalent analytical challenge in spatial-temporal data mining, the discovery of co-movement patterns involves identifying all groups of objects that move within spatial proximity for a specified time interval. Several co-movement patterns have been proposed, each with distinct constraints in terms of spatial proximity and temporal consecutiveness. For instance, the flock~\cite{DBLP:conf/gis/GudmundssonK06,DBLP:conf/gis/VieiraBT09} and group~\cite{DBLP:journals/dke/WangLH06} patterns require all objects in a group to be enclosed by a disk with a radius $r$. On the other hand, the convoy~\cite{DBLP:journals/pvldb/JeungYZJS08,DBLP:journals/pvldb/OrakzaiCP19,DBLP:conf/cikm/Liu0LLYW21}, swarm~\cite{DBLP:journals/pvldb/LiDHK10}, and platoon~\cite{DBLP:journals/dke/LiBK15} patterns adopt a more relaxed density-based spatial clustering approach to determine spatial proximity. Regarding the temporal dimension, flock and convoy patterns require global consecutiveness, where all timestamps of candidate groups must be consecutive. In contrast, group and platoon patterns tolerate local gaps between consecutive segments. 
 
 This co-movement pattern mining problem has found practical applications in animal behavior studies such as bird migration monitoring~\cite{DBLP:conf/gis/VieiraBT09}. In addition, government agencies can leverage the co-movement pattern mining task to support smart city management, such as detecting traffic congestion at varying levels of granularity~\cite{DBLP:journals/pvldb/OrakzaiCP19}, computing evacuation schedules during disaster~\cite{bhushan2017mining}, security surveillance upon suspicious groups~\cite{yadamjav2020querying}, and assessing the viability of launching public transport services in areas with dense co-movement patterns~\cite{DBLP:conf/cikm/Liu0LLYW21}. Prior works in this domain rely on large-scale GPS trajectories, which are primarily collected and owned by commercial hail-riding or map-service companies. However, in practice, these GPS-based trajectory data not directly accessible to government agencies. Therefore, in this paper, we migrate the data source from GPS sensors to surveillance cameras and present the first work on co-movement pattern mining against an urban-scale video database.  The idea is feasible because with the rapid development of smart cities, numerous surveillance cameras have been deployed throughout the road network. These infrastructures and the video big data they generate are owned by government agencies. Given the availability of computation-efficient and accurate visual inference models, these data sources can be leveraged to support the aforementioned applications of co-movement pattern mining.


Given a corpus of video data captured by the surveillance camera network, we can extract the movement paths of individual objects using trajectory recovery algorithms~\cite{DBLP:conf/gis/LinZH0WL21,DBLP:conf/mobicom/TongL0HH21,DBLP:conf/kdd/YuAYZWL22}. In contrast to GPS trajectories, the extracted path is represented by a sequence of camera identifiers and time intervals $(c_1,[t_1^s,t_1^e])\rightarrow(c_2,[t_2^s,t_2^e])\rightarrow\ldots\rightarrow(c_m,[t_m^s,t_m^e])$, where $c_i$ denotes the camera capturing the moving object and $[t_i^s,t_i^e]$ represents the period during which the object is detected by $c_i$. In this study, we assume that the trajectory extraction has been performed by a black-box algorithm and concentrate on devising efficient algorithms for mining co-movement patterns. Notably, the practice of treating video ingestion as a pre-processing step has also been adopted by recent video mining studies~\cite{DBLP:conf/sigmod/Chen0KY21,DBLP:journals/pvldb/ChenKYY22}.

In our scenario, the spatial location information outside the intervals of $[t_i^s,t_i^e]$ is not available. This property hinders the appliance of original trajectory-based co-movement patterns and calls for problem reformulation. To determine spatial proximity from a video database, we leverage temporal proximity at the same camera as an approximation. In other words, if two moving objects are captured by the same camera within time interval $\Delta_t$, we consider them spatially close to each other. As to the temporal duration for spatial proximity, we replace it with a sequence of $k$ cameras. In other words, we aim to identify platoons of objects that travel together along the same route in the road network.

Figure~\ref{fig:problem-example} illustrates a toy example. Three distinct objects, namely $o_1$, $o_2$, and $o_3$, traverse different routes on the road network. Specifically,  $o_1$ and $o_2$ follow the same trajectory, which is  $c_1\rightarrow c_2\rightarrow c_4\rightarrow c_5$, while $o_3$ travels from route $c_3\rightarrow c_2\rightarrow c_4\rightarrow c_5$. These objects are captured by the cameras at different time intervals. Suppose we set the minimum group size as $2$ and require that the minimum temporal gap at a camera to determine spatial proximity is $5$, and the objects must pass at least $3$ cameras together.  Based on these conditions, we find that the temporal gap between $o_1$ and $o_2$ is too large to form a valid co-movement pattern even though they share the same travel route. The minimum gap between their time intervals at camera $c_4$ is $7$. Conversely, $o_2$ and $o_3$ are spatially close to each other on the path $c_2\rightarrow c_4\rightarrow c_5$, and we can identify $\langle\{o_2,o_3\}, c_2\rightarrow c_4\rightarrow c_5\rangle$ as a valid co-movement pattern.

\begin{figure}[t!]
	\centering
		\includegraphics[width=0.42\textwidth]{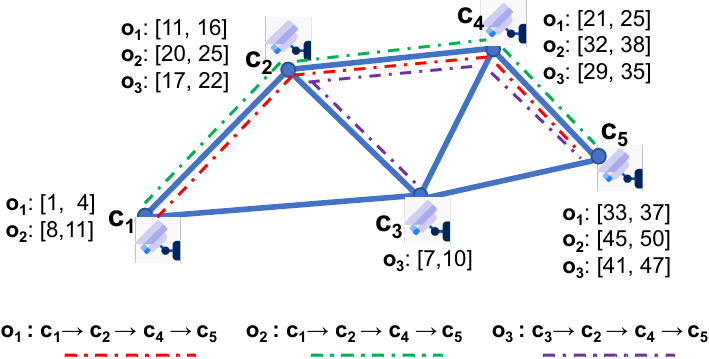}
  \vspace{-3mm}
	\caption{A toy example of co-movement pattern mining from surveillance cameras.}
	\label{fig:problem-example}
 \vspace{-3mm}
\end{figure}

In this paper, we formulate the problem of co-movement pattern mining and formally define the  spatial-temporal proximity constraints arising from video data.  This mining problem is challenging because of the potentially enormous search space that results from a large number of moving objects travelling over a prolonged  period across a road network with a plethora of surveillance cameras.  We establish the hardness of this problem by demonstrating that it is NP-Hard through the reduction of the maximum clique problem to our pattern mining problem.

Given the absence of readily applicable solutions, we draw upon the ideas from prior work and develop two baseline algorithms. The first algorithm is inspired by SPARE~\cite{DBLP:journals/pvldb/FanZWT16}, a parallel and general co-movement pattern miner from trajectories.  We adopt its Apriori-based enumerator to systematically explore the search space. The algorithm can be viewed as \textit{camera-ahead} exploration that enumerates candidate paths with increasing length. The remaining pattern constraints are then leveraged for further pruning. The second baseline extends the Coherent Moving Cluster (CMC) algorithm proposed in~\cite{DBLP:journals/pvldb/JeungYZJS08} for convoy pattern mining, with two tailored improvements for the video-based setting. The algorithm can be viewed as \textit{temporal-ahead} since it searches across the temporal dimension and dynamically updates the candidate patterns.

As the core technical contributions, we propose a new index called temporal-cluster suffix tree (TCS-tree), which performs two-level temporal clustering within each camera.  We represent the trajectory of each moving object as a sequence of cluster IDs and build a suffix tree to facilitate frequent subsequence mining. Based on TCS-tree, we propose a sequence-ahead pruning framework that allows for the simultaneous leverage of all pattern constraints to filter candidate travel paths. Finally, to minimize the verification cost of candidate paths, we propose a sliding-window-based co-movement pattern enumeration strategy and a hashing-based dominance eliminator, both of which effectively eliminate redundant operations.

In summary, we make the following contributions:
\begin{enumerate}
    \item This study pioneers the investigation of co-movement pattern mining from video databases. We formulate this new mining problem and theoretically prove its NP-hardness.
    \item  We propose a novel index called temporal-cluster suffix tree (TCS-tree), an efficient sequence-ahead pruning framework to facilitate pruning, as well as an effective verification scheme based on a sliding-window based candidate enumeration strategy and a hashing-based dominance eliminator.
    \item We conduct extensive experiments to validate the superiority of our proposed algorithm over the baselines. We also construct a video database with $1169$ cameras to  perform an end-to-end pipeline analysis and study the performance gap between GPS-driven and video-driven co-movement pattern miners. 
    \item To benefit the research community, we  make all datasets and implementation code available on github\footnote{\url{https://github.com/Mateng0228/Co-movement-Pattern-Mining-from-Videos}}. Furthermore, we identify several promising research directions for future work.
\end{enumerate}

The remainder of this paper is organized as follows. We formulate the new mining problem and theoretically prove its hardness in Section~\ref{sec:problem}. Section~\ref{sec:related-work} provides a survey of related works on co-movement pattern mining.  The two baseline algorithms are presented in Section~\ref{sec:baselines} and we propose our index and mining framework in Section~\ref{sec:tcs-tree}. In Section~\ref{sec:exp}, we conduct comprehensive experiments for performance evaluation. Finally, we conclude the paper and discuss potential avenues for future research in Section~\ref{sec:conclude}.

\section{Problem Description}\label{sec:problem}

\subsection{Data Model}
In our model for video database, a multitude of surveillance cameras is deployed throughout the urban city.  These cameras capture a vast amount of objects in motion over different time periods. For ease of presentation, we assume that the cameras are not overlapped and each moving person or vehicle can only be captured by one camera. The extension to support overlapped cameras will be discussed in Section 5.5. As the data source solely comprises of camera metadata (e.g., location, resolution, and view angle) and the recorded video content, GPS trajectories of moving objects are not available directly. Nonetheless, we can employ recent trajectory recovery algorithms~\cite{DBLP:conf/gis/LinZH0WL21,DBLP:conf/mobicom/TongL0HH21,DBLP:conf/kdd/YuAYZWL22} to extract the travel routes of moving objects across different cameras. As such, to support spatio-temporal pattern mining, we define the travel path of a moving object $o_i$ in the form of \blue{a} camera id sequence.

\begin{definition} Travel Path \\
We define the travel path of an object $o_i$ as a sequence of camera IDs, arranged in order of the time periods in which $o_i$ was captured: $P_i=(c_1,[s_1,e_1]) \rightarrow (c_2,[s_2,e_2]) \rightarrow \ldots \rightarrow (c_n,[s_n,e_n])$, where $s_j<e_j<s_{j+1}$ and $(c_j,[s_j,e_j])$ implies that $o_i$ appears in $c_j$ in the time interval $[s_j,e_j]$.
\end{definition}
We call $s_j$ the entrance timestamp and $e_j$ the exit timestamp. The time interval of $P_i$ is $[s_1,e_n]$. In cases without the need to mention the temporal dimension, we simplify the path representation as $P_i=c_1\rightarrow c_2\rightarrow \ldots \rightarrow c_n$. 

\begin{definition} Sub-path \\
We say $P_1$ is a sub-path of $P_2$, denoted by $P_1\subseteq P_2$, if the camera sequence of $P_1$ is a subsequence of $P_2$ and the time interval of $P_1$ is contained by $P_2$.
\end{definition}

\blue{
Here is an intuitive explanation regarding the travel path and 
the sub-path. In Figure~\ref{fig:problem-example}, object $o_3$'s travel path $P = c_3\rightarrow c_2\rightarrow c_4\rightarrow c_5$, whose time interval is $[7, 47]$. Notice that $o_3$ has a common sequence $P_{common} = c_2 \rightarrow c_4 \rightarrow c_5$ with other objects. Since $P_{common}$ is a continuous sub-sequence of $P$ and its time interval $[17, 47]$ is within $[7, 47]$, we call $P_{common}$ a sub-path of $P$.
}

From the travel paths of all objects, we can construct a camera network in a data-driven manner. Its vertices contain all disjoint cameras that appear in the travel paths. Two cameras $c_i$ and $c_j$ are connected in the camera network as long as there exists a travel path containing $c_i\rightarrow c_j$. The data model of camera network will be used in the two baseline algorithms.

\subsection{Mining Task Formulation}
To define co-movement patterns over the travel paths of moving objects, we notice that the traditional definition of co-movement patterns involves a parameter $\epsilon$ to determine spatial proximity, where $\epsilon$ is either the radius of the disk or the threshold for density reachability.  
However, in the scenario considered in this paper, the exact \blue{GPS} locations of moving objects are not available, and we must rely on other clues to determine spatial proximity.  In the following definition, we utilize temporal proximity at the sam camera to infer spatial proximity.

\begin{definition} $\epsilon$-reachability at camera $c_j$\\ Two objects $o_1$ and $o_2$ are $\epsilon$-reachable at camera $c_j$ if the gap between their entrance timestamp at  $c_j$ is no greater than $\epsilon$. 
\end{definition}

\blue{
In the example provided in Figure~\ref{fig:problem-example}, the gap between the entrance timestamps of $o_1$ and $o_2$ at camera $c_1$ is 8-1=7. Therefore, for any $\epsilon \geq 7$, both $o_1$ and $o_2$ are $\epsilon$-reachable at camera $c_1$.
}

The measurement of temporal duration, such as ``$k$ consecutive timestamps'' defined in convoy and flock patterns, is not applicable in our scenario because the temporal dimension of travel path $P_i$ only involves a sequence of disjoint time intervals. In our problem definition, we replace the constraint of temporal duration with ``$k$ consecutive cameras in the travel paths'' of the grouped objects.

\begin{definition}  Co-movement Pattern from Videos\\
Given parameters $\epsilon$, $m$ and $k$,  a co-movement pattern $CP$ in this paper is defined as $\langle O_i,P_i\rangle$, where $O_i$ is an object set and $P_i$ is a travel path. They satisfy the following constraints:
\begin{enumerate}
    \item $O_i$ contains at least $m$ objects, i.e., $|O_i|\geq m$.
    \item $P_i$ contains at least $k$ cameras, i.e., $|P_i|\geq k$.
    \item For each object $o\in O_i$ with travel path $P_o$, $P_i$ is a sub-path of $P_o$, i.e., $P_i\subseteq  P_o$.
    \item For any two objects $o_i,o_j\in O_i$ and camera $c_i\in P_i$,  $o_i$ and $o_j$ are $\epsilon$-reachable at camera $c_i$.
\end{enumerate}
\end{definition}


Under certain parameter settings, such as small values for $m$ and $k$ or a large value of $\epsilon$, it is possible that an extensive number of valid patterns may arise. In this paper, we follow previous approach~\cite{DBLP:journals/pvldb/OrakzaiCP19} to study  maximal co-movement pattern mining and define pattern dominance as following.
\begin{definition}  Pattern Dominance\\
A co-movement pattern $CP_1=\langle O_1,P_1\rangle$ is dominated by $CP_2=\langle O_2,P_2\rangle$ if $O_1\subseteq O_2$ and $P_1\subseteq P_2$.
\end{definition}
Our objective is to identify all the co-movement patterns $CP_i=\langle O_i,P_i\rangle$ that are both valid and non-dominated. The hardness of the problem is proved in the following theorem.


\begin{theorem}\label{thm:np-hard}
The problem of maximal co-movement pattern mining in video database is NP-hard.
\end{theorem}

\begin{proof}
We prove the theorem by reducing the maximum clique problem to our pattern mining problem.

Let $G=(V,E)$ be a graph with $n$ vertices. We conduct the reduction as follows. Each vertex $v_i\in G$ is uniquely mapped to a moving object $o_i$ in our problem. Let $\eta=\max(|V|, k)$ and we make all moving objects  pass through the same camera sequence $c_1\rightarrow c_2\rightarrow\ldots\rightarrow c_\eta$. The entrance time of object $o_i$ at camera $c_j$ is defined as
\begin{equation}
  t_{ij} =
    \begin{cases}
      (2\epsilon+1)*j-\epsilon & \text{$i\neq j\;\wedge\; $ $i$ and $j$ are not connected in $G$}\\
      (2\epsilon+1)*j & \text{\text{$i\neq j\;\wedge\; $ $i$ and $j$ are  connected in $G$}}\\
      (2\epsilon+1)*j+\epsilon & \text{$i=j$}
    \end{cases}       
\end{equation}
With the above settings, we know that if $v_a$ and $v_b$ are not connected in $G$, the entrance time of $o_a$ at camera $c_b$ is $(2\epsilon+1)*b-\epsilon$ and the entrance time of $o_b$ at camera $c_b$ is $(2\epsilon+1)*b+\epsilon$. These two objects are not $\epsilon$-reachable at camera $c_b$ and will not be co-located within a valid co-movement pattern. On the other hand, if $v_a$ and $v_b$ are connected in $G$, we show that $|t_{aj}-t_{bj}|\leq \epsilon$ for all cameras $c_1\rightarrow c_2\rightarrow\ldots\rightarrow c_\eta$: 
\begin{enumerate}[leftmargin=*]
    \item If $j=a$ or $j=b$, $|t_{aj}-t_{bj}|\leq |(2\epsilon+1)*j+\epsilon-(2\epsilon+1)*j|\leq \epsilon$. 
    \item If $j\neq a$ and $j\neq b$, $|t_{aj}-t_{bj}|\leq |(2\epsilon+1)*j-\epsilon-(2\epsilon+1)*j|\leq \epsilon$. 
\end{enumerate}
Therefore, if there exists a maximum clique with size $m$, we can also find a co-movement pattern covering $m$ objects through $\eta$ consecutive cameras. Obviously, this co-movement pattern is maximal.
\end{proof}

\blue{
\begin{example}
We use the example in Figure~\ref{fig:problem-example} to explain the results of co-movement patterns under different parameter settings. Suppose the query parameters are set to $m=2$, $k=3$, $\epsilon=12$, all the valid patterns that satisfy the pattern constraints include $CP_1:\langle\{o_1,o_2\}, c_1 \rightarrow c_2 \rightarrow c_4 \rangle$, $CP_2:\langle\{o_1,o_2\}, c_2 \rightarrow c_4 \rightarrow c_5 \rangle$, $CP_3:\langle\{o_1,o_3\}, c_2 \rightarrow c_4 \rightarrow c_5 \rangle$, $CP_4:\langle\{o_2,o_3\}, c_2 \rightarrow c_4 \rightarrow c_5 \rangle$, $CP_5:\langle\{o_1,o_2,o_3\}, c_2 \rightarrow c_4 \rightarrow c_5 \rangle$ and $CP_6:\langle\{o_1,o_2\}, c_1 \rightarrow c_2 \rightarrow c_4 \rightarrow c_5 \rangle$. After removing those dominated patterns, the results returned to the user include $CP_5:\langle\{o_1,o_2,o_3\}, c_2 \rightarrow c_4 \rightarrow c_5 \rangle$ and $CP_6:\langle\{o_1,o_2\}, c_1 \rightarrow c_2 \rightarrow c_4 \rightarrow c_5 \rangle$. 
If the user adjusts the parameter $\epsilon$ to $6$ with a tighter constraint on the $\epsilon$-reachability, the results set for the new query with $m=2$, $k=3$ and $\epsilon=6$ includes only one pattern: $\langle\{o_2,o_3\}, c_2 \rightarrow c_4 \rightarrow c_5 \rangle$.
\end{example}
}

\subsection{End-to-End Mining Pipeline}\label{sec:pipeline}
The complete mining pipeline that starts from the input of surveillance videos \blue{includes} three major steps:

\noindent \textbf{Step 1: Pre-processing}. The goal is to extract the cross-camera trajectories from the surveillance cameras deployed in a road network. This step can be achieved by adopting existing trajectory recovery algorithms from large-scale videos~\cite{DBLP:conf/gis/LinZH0WL21,DBLP:conf/mobicom/TongL0HH21,DBLP:conf/kdd/YuAYZWL22} or multi-camera multi-target tracking models~\cite{DBLP:conf/cvpr/TangN0YBWKAH19,DBLP:conf/cvpr/HsuHWCLH19,DBLP:conf/cvpr/HeHYHWG20}. The latter is more accurate but consumes higher computation cost. The trajectory extraction can be viewed as a pre-processing step and the extracted trajectories can be used to support co-movement pattern mining or video database queries. 

\noindent \textbf{Step 2: Online Index Construction}. After extracting the camera sequences for all moving objects, we can build index in an online manner to facilitate co-movement pattern mining. For example, we will build a temporal-clustering suffix tree, which relies on the query threshold $\epsilon$.

\noindent \textbf{Step 3: Pattern Mining}. Users can apply the mining algorithms proposed \blue{in} this paper to retrieve all the co-movement patterns satisfying the spatial-temporal proximity constraints.
\blue{For parameter setting, $\epsilon$ can be set from $1$-$4$ minutes to take into account the effect of traffic lights. Parameters $m$ and $k$ can be determined by specific applications.}
Since the pattern mining is an exploratory process, if the query user is not satisfied with the results, he/she can adjust the parameters and proceed to step 2 for the next round of pattern mining.

For ease of reference, the notations frequently used in the paper are summarized in Table~\ref{table:notation}.
\begin{table}[h!]
\begin{center}
\vspace{-2mm}
\caption{Notation Table.}
\vspace{-2mm}
\small
\begin{tabular}{|c|p{6.5cm}|} \hline
$c_i$, $o_i$ & A surveillance camera $c_i$ and a moving object $o_i$\\ \hline
$P_i$ & The travel path of object $o_i$ \\ \hline
$O_i$ & An object set  \\ \hline
$\epsilon$ & The threshold for temporal closeness \\ \hline
$m$ & The minimum group size for co-movement pattern \\ \hline
$k$ & The minimum route length for co-movement pattern \\ \hline
$CP$ & A co-movement pattern \\ \hline
$TC$ & A temporal cluster represented as triplet $(O,l,r)$\\ \hline
 $\mathbb{TC}_{c_i}$ & The temporal clusters in camera $c_i$ \\ \hline
\end{tabular}
\label{table:notation}
\end{center}
\vspace{-5mm}
\end{table}

\section{Related Work}\label{sec:related-work}

Among the co-movement patterns defined over GPS trajectories, the flock~\cite{DBLP:conf/gis/GudmundssonK06,DBLP:conf/gis/VieiraBT09} and  convoy~\cite{DBLP:journals/pvldb/JeungYZJS08,DBLP:journals/pvldb/OrakzaiCP19,DBLP:conf/cikm/Liu0LLYW21} patterns require the candidate groups to appear in $k$ consecutive timestamps. Their difference lies in the definition of spatial proximity. Flock requires objects in the same cluster to be within a disk-region of diameter less than parameter $\epsilon$, while convoy uses density-based spatial clustering. In the mining algorithm design, the first step is to perform clustering at each timestamp. Vieira et al. employed the grid index approach to identify disks~\cite{DBLP:conf/gis/VieiraBT09}, while convoy pattern mining relies on density-based clustering. Given that clustering is often computationally expensive and can become a performance bottleneck, several techniques have been proposed to alleviate this overhead.  Jeung et al. adopt trajectory simplification to convert trajectories into segments and perform segment clustering to reduce computation overhead.  \cite{DBLP:conf/cikm/Liu0LLYW21} improves clustering efficiency with the help of spatial partitioning. Alternatively, Orakzai et al. proposed the $k$/2-hop algorithm, which performs DBSCAN only at a small subset of timestamps, thereby quickly pruning objects that have no potential to form a convoy~\cite{DBLP:journals/pvldb/OrakzaiCP19}.

Different from flock and convoy, the group~\cite{DBLP:journals/dke/WangLH06}, swarm~\cite{DBLP:journals/pvldb/LiDHK10} and platoon~\cite{DBLP:journals/dke/LiBK15}  patterns have more relaxed constraints on the pattern duration. Their techniques of mining are of the same skeleton. The main idea of mining is to grow an object set from an empty set in a depth-first manner. Throughout this construction process, various pruning techniques are employed to discard unnecessary branches.  Group pattern mining~\cite{DBLP:journals/dke/WangLH06} uses a VGgraph to guide the pruning of false candidates, while swarm mining~\cite{DBLP:journals/pvldb/LiDHK10}  relies on two more pruning rules called backward pruning and forward pruning. The platoon mining algorithm in~\cite{DBLP:journals/dke/LiBK15} takes advantage of a prefix table structure to steer the depth-first search.

There also exist several parallel co-movement pattern mining algorithms that adopt distributed computing for performance boosting. In~\cite{DBLP:conf/gis/OrakzaiDC15,DBLP:conf/mdm/OrakzaiCP16}, Orakzai et al. adopt MapRuduce framework for distributed convoy pattern mining. In~\cite{DBLP:journals/pvldb/FanZWT16}, a general and parallel framework called SPARE is proposed to handle multiple types of co-movement patterns based on Spark. SPARE achieves workload balance by partitioning objects into fine granular stars. For each partition, an Apriori-based Enumerator is adopted to mine the co-movement patterns. \cite{DBLP:journals/geoinformatica/OrakzaiPC21} adopts a divide and conquer strategy for distributed convoy pattern mining. Each node runs DBSCAN on a partition and mines local convoys, which are then combined to produce the global result. 

The differences between GPS-based and video-based co-movement pattern mining are summarized as shown in Table~\ref{table:comp-gps-video}.
\begin{table}[h!]
\setlength{\abovecaptionskip}{2mm}
\begin{center}
\caption{Differences between GPS-based and video based co-movement pattern mining.}
\label{table:comp-gps-video}
\small
\begin{tabular}{|c|p{2cm}|p{2.5cm}|} \hline
 & Trajectory-based  & Video-based \\ \hline
Element in trajectory & $\langle lat, lng, time\rangle$  & $\langle c_j,[s_j,e_j]\rangle$ \\ \hline
Trajectory representation & synchronized and fine-grained  & asynchronized and coarse-grained\\ \hline
Proximity Determination & spatial clustering  & temporal clustering within a camera \\ \hline
Duration Determination & consecutive timestamps & $k$ consecutive paths in the camera network \\ \hline
Performance bottleneck & spatial clustering & candidate verification \\ \hline
\end{tabular}
\end{center}
\vspace{-2mm}
\end{table}

For co-moving pattern mining based on trajectory streaming,  CoMing~\cite{DBLP:journals/pvldb/ChenGFMJG19,DBLP:conf/sigmod/FangGPCMJ20} leverages Apache Flink for streaming-based processing, builds spatial indexes to accelerate clustering, and employs pruning techniques to eliminate unnecessary verification. In~\cite{DBLP:journals/kais/SheinPI20}, the group pattern is relaxed to allow membership withdrawal or re-join as long as some participators stay connected for all time intervals. An incremental discovery solution is developed to retrieve the evolving companion efficiently from trajectory streaming data.

Recently, video-based pattern mining has emerged as an interesting database research topic. In~\cite{DBLP:conf/gis/LinZH0WL21,DBLP:conf/mobicom/TongL0HH21,DBLP:conf/kdd/YuAYZWL22}, the task of snapshot clustering is studied to recover trajectories from urban-scale videos. The objective is to cluster snapshots of detected objects that are visually similar and spatially coherent. The items within each cluster are considered to belong to the same vehicle.  They are sorted by the associated timestamp to form a trajectory. In~\cite{DBLP:conf/sigmod/Chen0KY21,DBLP:journals/pvldb/ChenKYY22}, Chen et al. retrieve video segments that satisfy user-provided constraints on the spatial or temporal relationships among objects of interest, as well as other conditions on the object labels.

\section{Two Baseline Algorithms}\label{sec:baselines}
In this section, we present two baseline algorithms for co-movement pattern mining, by adapting the ideas from trajectory-based solutions to the context of video database.

\subsection{CMC Algorithm}
The Coherent Moving Cluster (CMC) algorithm~\cite{DBLP:journals/pvldb/JeungYZJS08,DBLP:conf/gis/OrakzaiDC15} represents a simple and general baseline technique for mining convoy patterns. The algorithm operates by performing spatial clustering on the moving objects at each timestamp, followed by a search for all valid patterns across the temporal dimension. To this end, a global buffer of candidate patterns is maintained. At each timestamp, each cluster is intersected with the partial patterns in the global buffer, thus generating new patterns that adhere to the constraints of $k$ consecutive timestamps and minimum group size.

In this paper, we present a competitive baseline, which builds upon CMC with two key improvements tailored for video-based \blue{settings}. Firstly, we substitute the traditional spatial clustering approach with a temporal clustering approach that ensures objects within the same cluster are $\epsilon$-reachable. The temporal clustering is performed independently within each camera.  Secondly, we replace the global buffer with distributed buffers across the cameras. At each timestamp, when accessing a temporal cluster within camera $c_i$, we intersect it with only the patterns in the local buffer of $c_i$, thus reducing the search space and enhancing computational efficiency. The newly generated patterns are  propagated to the buffers in the neighboring cameras for future exploration.

\subsubsection{Temporal Clustering}
Recall that we utilize $\epsilon$-reachability  in the temporal dimension as an approximate means of representing spatial proximity. For the set of objects captured by camera $c_i$, we propose a lightweight temporal clustering approach that ensures objects within the same cluster are $\epsilon$-reachable. Additionally, its linear complexity can prevent the clustering process from becoming a performance bottleneck, 

As demonstrated in Algorithm~\ref{alg:clustering}, we access the moving objects in ascending order of their entrance time into camera $c_i$, utilizing two anchors, $l$ and $r$ to establish the boundaries of each cluster.   To store the resulting clusters of camera $c_i$, we maintain a data structure $\mathbb{TC}_{c_i}=\{TC_1,TC_2,\ldots,TC_i,\ldots\}$, where each temporal cluster $TC_i=\left(O_i,\left[l_i, r_i\right]\right)$ is associated with an object set $O_i$ and time interval $[l_i,r_i]$. By employing the anchor $r$ to determine the right boundary, all moving objects within temporal distance $\epsilon$ are grouped into a single cluster. We exclude small clusters that contain fewer than $m$ objects.

\begin{algorithm}[t!]
\SetAlgoNoEnd \SetAlgoNoLine
\caption{Temporal Clustering}\label{alg:clustering}
\small
$\mathbb{TC}_{c_i} \leftarrow \emptyset;$ $TC \leftarrow (\emptyset, [0,0]);$ $r\leftarrow 1;$ $l\leftarrow 1$\;
$S \leftarrow$ object records in order of entrance time in camera $c_i$\;
\While{$r \leq |S|$}{
    \eIf{$S[r].time - S[l].time \leq \epsilon$}{
		Enqueue $S[r].obj$ into $TC.O$\; 		
		$TC.r \leftarrow S[r].time;$ $r++$\;
    }{
	\If{$|TC|\geq m$}{ $\mathbb{TC}_{c_i} \leftarrow \mathbb{TC}_{c_i} \bigcup TC$}
	\While{$l\neq r$}{
		\If{$S[r].time - S[l].time \leq \epsilon$}{
			\textbf{break}\;			
		}{
            Dequeue $S[l].obj$ from $TC.O$\;
			$l++;$ $TC.l \leftarrow S[l].time$
		}		
	}
    }
}
\lIf{$|TC|\geq m$}{ $\mathbb{TC}_{c_i} \leftarrow \mathbb{TC}_{c_i} \bigcup TC$}
\textbf{return} $\mathbb{TC}_{c_i}$\;
\end{algorithm}

\begin{figure}[h!]
	\centering
		\includegraphics[width=0.42\textwidth]{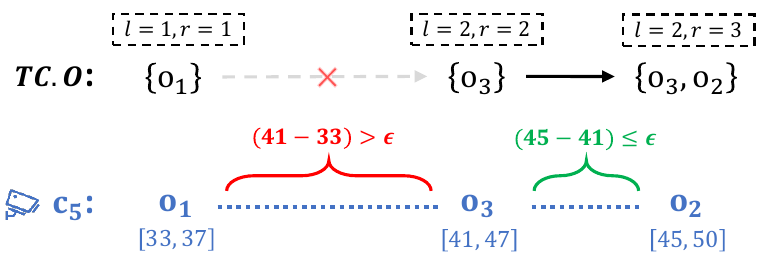}
	\caption{A toy example of Temporal Clustering.}
	\label{fig:index-graph-tc}
\end{figure}
\blue{
\begin{example}
Figure~\ref{fig:index-graph-tc} presents a toy example illustrating the execution process of Algorithm~\ref{alg:clustering} on objects passing through camera $c_5$ in Figure~\ref{fig:problem-example} (assuming $m=2$, $\epsilon=6$).
The algorithm initially obtains object records in the order of their entrance time into camera $c_5$, namely, objects $o_1$, $o_3$, and $o_2$.
Subsequently, each object undergoes processing by adjusting the anchor $r$. When $r=1$, the first object $o_1$ is added to $TC.O$, and the corresponding time interval is updated. When $r=2$, the next object $o_3$ is processed. Since the time gap between the entrance times of $o_3$ and $o_1$ exceeds $\epsilon$, $o_1$ cannot continue to be retained in $TC.O$. The algorithm removes the information related to $o_1$ from $TC$, shifts the anchor $l$, and simultaneously incorporates the relevant information of $o_3$ into $TC$. When $r=3$, the final object $o_2$ is processed. It is observed that $o_2$ and $o_3$ are $\epsilon$-reachable. Therefore, $o_3$ is included in $TC$ while retaining $o_2$. Finally, since the number of objects in $TC.O$ is greater than $m$, $TC$ is added to the result $\mathbb{TC}_{c_5}$, thereby completing the algorithm.
\end{example}
}

Upon examining the temporal clustering, we can draw three observations. First, as each object is accessed by the anchor points at most twice, the time complexity of temporal clustering is linear to the number of objects captured by camera $c_i$. Assuming there are $N$ objects, each with an average path length of $L$, the complexity of temporal clustering on all cameras is $O(NL)$. Second, any two objects within the same temporal cluster are guaranteed to be $\epsilon$-reachable, as specified by line $5$ in Algorithm~\ref{alg:clustering}. Finally, these temporal clusters are not necessarily disjoint in the temporal dimension. Thus, it is possible that an object in a camera belongs to multiple temporal clusters.

\subsubsection{Candidate Pattern Enumeration}
The search algorithm explores the temporal dimension for candidate pattern enumeration. We sort the temporal clusters $TC_i=\left(O_i,\left[l_i, r_i\right]\right)$ from all cameras in ascending order of their starting time $l_i$. Let $V_{in}^i$ denote the incoming neighbors of camera $c_i$ in the camera network and $V_{out}^i$ denote its outgoing neighbors. We maintain a local buffer for each camera to store candidate patterns whose path ends at any node in $V_{in}^i$. Suppose we are presently examining the temporal cluster $TC_i=\left(O_i,\left[l_i, r_i\right]\right)$ derived from camera $c_u$, we perform object intersection to generate new candidate patterns. This step is similar to the original CMC algorithm which expands the time interval of a pattern from $t$ timestamps to $t+1$ timestamps. For each candidate pattern $\langle O_b, P_b\rangle$ in the local buffer of camera $c_u$, we examine the size of $O_b\bigcap O_i$. If $|O_b\bigcap O_i|\geq m$, we generate a new pattern $\langle O_b\bigcap O_i, P_b\rightarrow c_u\rangle$ and propagate it to the outgoing neighbor cameras in  $V_{out}^u$. If this is a valid pattern, we also push it into result set $\mathbb{R}$. The algorithm terminates once all temporal clusters have been evaluated. Finally, we remove the dominated patterns in $\mathbb{R}$ by building \blue{an} inverted index to quickly examine the subset relationship between object sets.

\begin{algorithm}[t!]
\SetAlgoNoEnd  \SetAlgoNoLine
\caption{CMC Algorithm}\label{alg:cmc}
\small
$TC\leftarrow$ Sort clusters $TC_i=\left(O_i,\left[l_i, r_i\right]\right)$ from all cameras by $l_i$\;
\For{each $TC_i$ associated with camera $c_u$}{
    \For{each candidate pattern $\langle O_b,P_b \rangle \in Buf_{c_u}$}{
        \If{$|O_b \bigcap O_i|\geq m$}{
            \If{$|P_b| \geq k-1$}{
                $\mathbb{R} \leftarrow \mathbb{R} \bigcup\; \langle O_b \bigcap O_i, P_b \rightarrow c_u \rangle$\;
            }
            \For{each camera $c_{o} \in V_{out}^{u}$ }{
                $Buf_{c_{o}} \leftarrow Buf_{c_{o}} \bigcup\; \langle O_b \bigcap O_i, P_b \rightarrow c_u \rangle$\;
            }
        }
    }
        \For{each camera $c_{o} \in V_{out}^{c_u}$ }{
                $Buf_{c_{o}} \leftarrow Buf_{c_{o}} \bigcup\; \langle O_i, c_u \rangle$\;
            }
}
Remove dominated patterns in $\mathbb{R}$\;
\textbf{return} $\mathbb{R}$\;
\end{algorithm}

\begin{figure}[h!]
	\centering
		\includegraphics[width=0.45\textwidth]{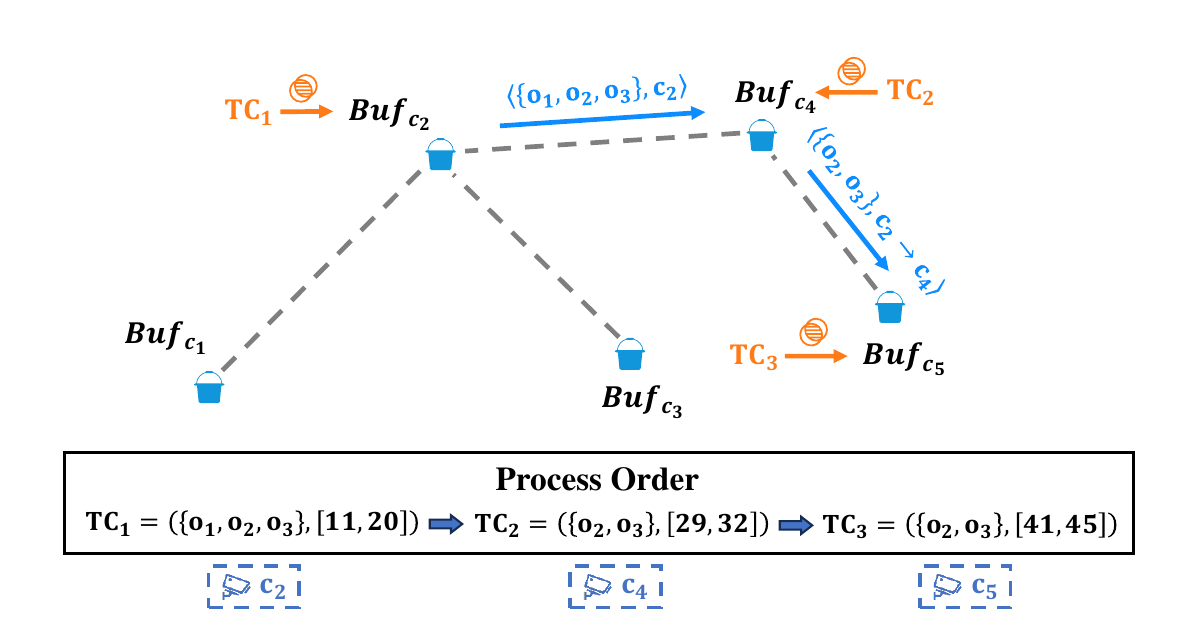}
	\caption{An example of CMC Algorithm.}
	\label{fig:index-graph-cmc}
\end{figure}
\blue{
\begin{example}
We use an example to briefly illustrate our CMC Algorithm. Figure~\ref{fig:index-graph-cmc} explains how Algorithm~\ref{alg:cmc} operates on the toy data model presented in Figure~\ref{fig:problem-example} (assuming $m=2$, $k=3$ and $\epsilon=6$).
After applying the Temporal Clustering algorithm, we can get three temporal clusters: $TC_1=(\{o_1,o_2,o_3\},[11,20])$ for camera $c_2$, $TC_2=(\{o_2,o_3\},[29,32])$ for camera $c_4$, and $TC_3=(\{o_2,o_3\},[41,45])$ for camera $c_5$. We first construct a camera network using all travel paths and assign a local buffer $Buf_{c_u}$ to each camera $c_u$.
Then, Algorithm~\ref{alg:cmc} can be executed to process each temporal cluster in ascending order of starting time. Specifically, we begin by processing $TC_1$ that belongs to camera $c_2$, where we take the intersection with elements in $Buf_{c_2}$. Since the result is an empty set, only $\langle \{o_1,o_2,o_3\}, c_2 \rangle$ generated from $TC_1$ is forwarded to the next camera $c_4$'s local buffer $Buf_{c_4}$. For $TC_2$, the intersection of $TC_2$ and $Buf_{c_4}$ yields $\{o_2,o_3\}$, allowing us to extend $\langle \{o_1,o_2,o_3\}, c_2 \rangle$ to $\langle \{o_2,o_3\}, c_2 \rightarrow c_4 \rangle$ and pass it to $Buf_{c_5}$. Finally, we process $TC_3$, taking the intersection with recently updated $Buf_{c_5}$, resulting in the final result $\langle \{o_2,o_3\}, c_2 \rightarrow c_4 \rightarrow c_5 \rangle$.
\end{example}
}

\subsection{Apriori Based Enumerator}
Our second baseline is inspired by SPARE~\cite{DBLP:journals/pvldb/FanZWT16}, which is a parallel and general framework to mine various co-movement patterns. Within SPARE, an Apriori-based enumerator is introduced to facilitate space pruning. To adopt the idea, we construct the combinatorial space using cameras because we can leverage the ontology of camera network to avoid enumerating all the combinations of cameras. In addition, it is straightforward that the anti-monotonicity property also holds: 
\begin{property}
Let $\mathbb{O}_i$ be the set of $\epsilon$-reachable objects in all the cameras of $P_i$. If $P_1$ is a sub-path of $P_2$, then $\mathbb{O}_1\supseteq \mathbb{O}_2$.
\end{property}
\noindent With the property, if we cannot find a group of objects with size $m$ that are $\epsilon$-reachable in all cameras of $P_i$, then all the expanded paths from $P_i$ can be pruned.

We present our Apriori-based enumerator in Algorithm~\ref{alg:apriori}. The general idea is to enumerate candidate paths in increasing number of cameras,  evaluate their validity by assessing the constraints related to object size $m$ and spatial proximity $\epsilon$, and eliminate any illegal or dominated outcomes. The algorithm starts from the first layer to construct candidate patterns with only one camera. In the initialization stage (lines $1$-$4$), temporal clustering via Algorithm~\ref{alg:clustering} is employed to perform temporal clustering within each camera $c_i$.  As long as its output $\mathbb{TC}_{c_i}$ is not empty, the camera is considered as a candidate path with $len=1$ and pushed into a buffer $Q$. These candidate cameras constitute the bottom layer of the lattice for Apriori-based enumeration. 

In the subsequent layers, we adhere to the Apriori algorithm's procedure to enumerate candidate paths of length $n+1$ derived from their two \blue{sub-paths} of length $n$. Given  $P=c_1\rightarrow c_2\rightarrow \ldots \rightarrow c_{n}\rightarrow c_{u}$, we impose the condition that its two \blue{sub-paths}, $P_1=c_1\rightarrow c_2\rightarrow \ldots \rightarrow c_{n}$ and $P_2=c_2\rightarrow \ldots \rightarrow c_{n}\rightarrow c_{u}$, must be present in the output of the previous layer (lines $6$-$9$). If either of these \blue{sub-paths} is absent, path $P$ can be pruned in accordance with the anti-monotonicity property. We construct $\mathbb{O}_{new}$ to maintain a group of valid object sets that appear in both $P_1$ and $P_2$, i.e., we intersect each pair of object set $O_1\bigcap O_2$ and discard it if its size is smaller than $m$ or it is dominated by other object set (lines $11$-$13$). If $\mathbb{O}_{new}$ is not empty, we find a new candidate path $P_{new}=c_1\rightarrow c_2\rightarrow \ldots \rightarrow c_{n}\rightarrow c_{u}$ and push $(P_{new},\mathbb{O}_{new})$ into $Q$. If $|P|\geq k$, we also add the valid candidate patterns into result set $\mathbb{R}$. When the buffer $Q$ becomes empty, the Apriori-based enumerator terminates and we remove dominated patterns in $\mathbb{R}$. This step is the same as the one employed in the CMC algorithm.




\begin{algorithm}[h!]
\SetAlgoNoEnd  \SetAlgoNoLine
\caption{Apriori-Based Enumerator}
\small
\label{alg:apriori}
\For{each camera $c_i$}{
    Perform temporal clustering  in $c_i$ using Algorithm~\ref{alg:clustering}\;
    \small
    $P_i\leftarrow c_i;$ $\mathbb{O}_i\leftarrow\mathbb{TC}_{c_i}$\;
    Add $(P_i, \mathbb{O}_i)$ into buffer $Q$\;
}
\While{$Q$ is not empty}{
    \For{each path $P_1=c_1\rightarrow c_2\rightarrow\ldots\rightarrow c_n\in Q$}{
         \For{each $c_o\in V_{out}^{n}$}{
            Let $P_2= c_2\rightarrow c_3\rightarrow \ldots\rightarrow c_n\rightarrow c_o$\;
            \If{$Q$ contains $P_2$}{
                $\mathbb{O}_{new}\leftarrow \{O_1\bigcap O_2|O_1\in \mathbb{O}_1, O_2\in\mathbb{O}_2\}$\;
                \If{$|O_i|<m$ or $O_i\subseteq O_j\in \mathbb{O}_{new}$}{
                    Remove $O_i$ from $\mathbb{O}_{new}$  \;
                }
                \If{$|\mathbb{O}_{new}|>0$}{
                    $P_{new}= P_1\rightarrow c_o$\;
                    $Q'\leftarrow Q'\bigcup \;(P_{new},\mathbb{O}_{new})$\; 
                    \If{$|P_{new}|\geq k$}{
                        \For{each $O_i\in \mathbb{O}_{new}$}{
                            $\mathbb{R}\leftarrow \mathbb{R}\bigcup\; \langle P_{new},O_i\rangle $\;
                        }
                    }
                }
            }
         }
   }
  $Q\leftarrow Q'$;  $Q'\leftarrow \emptyset$\;
}
Remove dominated patterns in $\mathbb{R}$\;
\textbf{return} $\mathbb{R}$\;
\end{algorithm}

\blue{
\begin{example}
We present an example in Figure~\ref{fig:apriori-example} to explain the procedure of our Apriori-Based Enumerator. We still use the data model  in Figure~\ref{fig:problem-example} and the query parameters are set to $m=2$, $k=3$, and $\epsilon=6$.
We begin by executing Algorithm~\ref{alg:clustering} to obtain all temporal clusters, including $(\{o_1,o_2,o_3\},[11,20])$ from camera $c_2$, $(\{o_2,o_3\},[29,32])$ from camera $c_4$, and $(\{o_2,o_3\},[41,45])$ from camera $c_5$. We treat these cameras as candidate paths of length 1, forming the bottom layer of the lattice for Apriori-based enumeration. Subsequently, we use the topology of the camera network and the anti-monotonicity property to iteratively concatenate candidate paths starting from the bottom layer. Specifically, as the candidate paths $c_2$ and $c_4$ are adjacent within the camera network and their shared object set $\{o_2,o_3\}$ satisfies $m$, we concatenate $c_2$ and $c_4$ to create the new path $c_2 \rightarrow c_4$. Likewise, candidate path $c_4 \rightarrow c_5$ can also be generated from the bottom layer. For the candidate paths $c_2 \rightarrow c_4$ and $c_4 \rightarrow c_5$ in the second layer, it's not difficult to recognize that they can be linked in terms of the path and their intersection of corresponding object sets results in $\{o_2,o_3\}$. Therefore, we connect them to form $c_2 \rightarrow c_4 \rightarrow c_5$.
Since there is only one path ($c_2 \rightarrow c_4 \rightarrow c_5$) in the third layer, we cannot generate any paths of length $4$ and the  enumeration process is terminated. Finally, we obtain the pattern $\langle \{o_2,o_3\}, c_2 \rightarrow c_4 \rightarrow c_5 \rangle$ corresponding to the path $c_2 \rightarrow c_4 \rightarrow c_5$.
\end{example}
}

\begin{figure}[h!]
	\centering
		\includegraphics[width=0.45\textwidth]{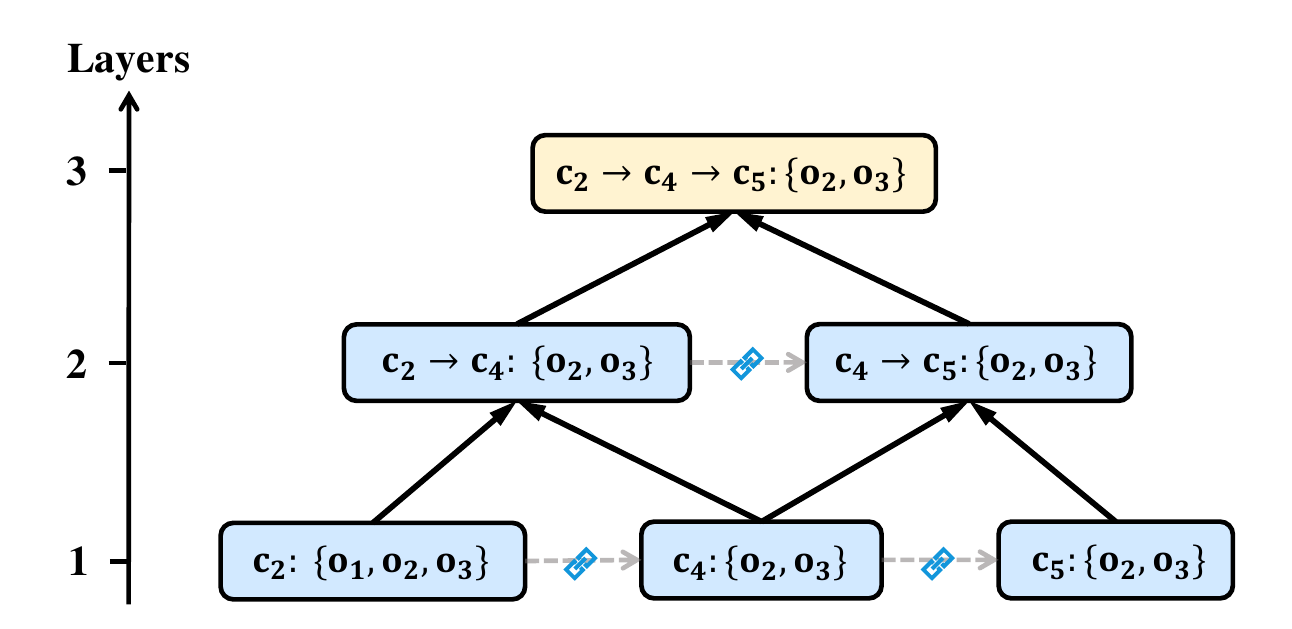}
	\caption{An example of candidate enumeration in Algorithm~\ref{alg:apriori}.}
	\label{fig:apriori-example}
\end{figure}

\section{Sequence-Ahead Mining Framework}~\label{sec:tcs-tree}
In this section, we first deliver the basic idea of sequence-ahead mining framework and the construction of temporal-cluster suffix tree (TCS-tree) to produce  co-movement pattern candidates. Among the candidates, we also propose an efficient sliding-window based  verification algorithm to remove false positives and a hashing-based dominance elimination strategy to retain maximal patterns. \blue{At the end of the section, we discuss the extended support of the current data model for the scenario of overlapped cameras.}

\subsection{Basic Idea}\label{sec:base-suffix-tree}
A co-movement pattern involves three constraints that allow us to devise pruning rules, including the $\epsilon$-reachable constraint for spatial proximity, the minimum group size $m$ and the minimum path length $k$. The two baseline algorithms adopt different strategies to construct the search space for candidate enumeration. More specifically, the Apriori-based enumerator can be viewed as camera-ahead exploration and its search space consists of all the possible paths with length $k$. During the candidate enumeration, the constraints on parameters $\epsilon$ and $m$ are further utilized for pruning. The CMC algorithm is temporal-ahead because it explores the search space from the temporal dimension. For each appearance of an object within a camera, it joins with the candidate patterns that are $\epsilon$-reachable. The validness is further verified using constraints on parameters $m$ and $k$.

In this paper, we propose a sequence-ahead mining framework with the objective of utilizing multiple constraints for candidate enumeration. A straightforward solution is to model each object $o_i$ as a sequence of camera ids in ascending order of the entrance time into the camera. Then, we can utilize the constraints of minimum group size $m$ and minimum path length $k$ simultaneously, and perform frequent subsequence mining (FSM) for candidate generation. A subsequence is a candidate path only if it satisfies two conditions, namely its length is at least $k$, and its support is at least $m$. Since FSM is a fundamental data mining problem that has been adequately addressed, we can readily adopt existing frequent subsequence miners for candidate enumeration. For example, we can construct a suffix-tree for the sequences and use the index to identify frequent subsequences. For each candidate path, we can further apply the constraint of $\epsilon$-reachability to eliminate false positives.  Finally, the patterns that are not dominated will be returned to the query user as mining results.



The aforementioned solution still has great room for improvement. In the following, we first present the concept of meta-clusters and propose a new index called temporal-clustering suffix tree (TCS-tree) that intends to integrate the three constraints for filtering in the frequent subsequence mining process. To reduce redundant verification overhead, we also propose a sliding-window based candidate enumeration strategy and a hashing based dominance eliminator.

\subsection{TCS-tree Construction}
In order to further refine the efficacy of the sequence-ahead mining framework and minimize the occurrence of false positive candidates, we propose an innovative strategy that involves extending the sequence of camera ids to incorporate temporal information and generate a more nuanced representation. The original sequence is constructed using the one-to-one mapping function $F(o_i)=c_j$, with each element in the sequence representing a camera id. However, this representation is limited in that it is unable to differentiate between two distinct travel paths that share the same sequence of camera ids but possess substantially different starting times. To address this limitation, we intend to partition the temporal dimension in a camera $c_j$ into multiple disjoint intervals $I_j^1, I_j^2\ldots$, and extend $F(o_i)$ to a new one-to-one mapping $G(o_i)=I_j^t$, thereby representing each object as a sequence of $I_j^t$. This approach enables us to incorporate the three constraints of $\epsilon$-reachability, minimum group size $m$, and minimum path length $k$ during the process of frequent subsequence mining.
%
%
Notably, it is possible for an object $o_i$ to be captured at multiple time periods at the same camera $c_j$. In this case, we can simply use $o_i^1$, $o_i^2,\ldots$ to represent different occurrence periods of $o_i$ such that each $o_i^u$ becomes the input of mapping function $G(\cdot)$. For ease of presentation, we still use $G(o_i)$ as the default notation unless further clarification is required.

It is obvious that smaller intervals are more powerful to distinguish false positives, i.e., if two objects share a travel path $P_i=c_1\rightarrow \ldots\rightarrow c_n$ but they are not $\epsilon$-reachable, it is more likely for these two objects to be projected into different intervals.  Therefore, we formulate the construction of function $G(o_i)$ as minimizing the total size of all intervals $\sum_{j,t} |I_j^t|$, with the following two constraints:
\begin{enumerate}
    \item $G(o_i)$ is a one-to-one mapping in camera $c_j$, i.e., each object $o_i$ is uniquely mapped to an interval $I_j^t$ at camera $c_j$.
   \item If $o_u$ and $o_v$ are $\epsilon$-reachable in camera $c_j$, they are mapped to the same $I_j^t$. This requirement guarantees  that there is no missing valid pattern.
\end{enumerate}

Note that we cannot directly apply the temporal clustering algorithm in Algorithm~\ref{alg:clustering} because it is possible that an object belongs to multiple clusters. This violates the constraint of one-to-one mapping requirement for $G(o_i)$ and we cannot represent an object with a sequence of cluster ids and then apply frequent subsequence mining. To resolve the issue, we propose a two-level temporal clustering algorithm and represent each object as a sequence of cluster ids, based on which we build a temporal cluster suffix tree (TCS-tree) for frequent subsequence mining.

In the first level of clustering, we apply Algorithm~\ref{alg:clustering} to generate temporal clusters that are $\epsilon$-reachable. In the second level, we progressively merge the temporal clusters that are overlapped. Two temporal clusters $TC_i$ and $TC_j$ are merged into a meta-cluster as long as they contain common objects, i.e., $TC_i.O_i\bigcap TC_j.O_j\neq \emptyset$. In implementation, we can scan the temporal clusters along the temporal dimension and treat the first cluster as a meta-cluster. When we access a temporal cluster $TC_i$, it is merged into previous meta-cluster $MC$ if they are overlapped. Otherwise,  $TC_i$ and $MC$ are disjoint and we treat $TC_i$ as a new meta-cluster to merge its subsequent temporal clusters. The process terminates when all clusters have been examined.

\begin{example}
We use a toy example to explain the results of meta-clusters and the construction of TCS-tree. As depicted in Figure~\ref{fig:index-graph-tcs}, there are $5$ objects that move in a road network with cameras $c_1$ to $c_6$. Each object is captured by part of these cameras.  To construct TCS-tree, we can perform temporal clustering within each camera and merge the overlapped clusters into a meta-cluster. For example, in camera $c_5$, there are $3$ temporal clusters $\{o_1,o_2,o_3\},\{o_3,o_4\},\{o_5\}$ and we can merge them into two meta-clusters $\{o_1,o_2,o_3,o_4\},\{o_5\}$. Each object is uniquely mapped to a meta-cluster $MC_j^i$ in camera $c_i$. Therefore, as shown in the figure, we can construct a new data model for each object and represent it as a sequence of meta-clusters. Then, we can build a TCS-tree on top of these sequences whose elements are meta-clusters. TCS-tree is essentially a suffix tree. Each unique suffix in the meta-cluster sequences is stored as a single node in a suffix tree. Each leaf node contains the starting position of the suffix it represents. Originally, there are $8$ frequent subsequences. With the fine-grained representation, the number is reduced to $4$.
\end{example}
\begin{figure}[h!]
	\centering
		\includegraphics[width=0.47\textwidth]{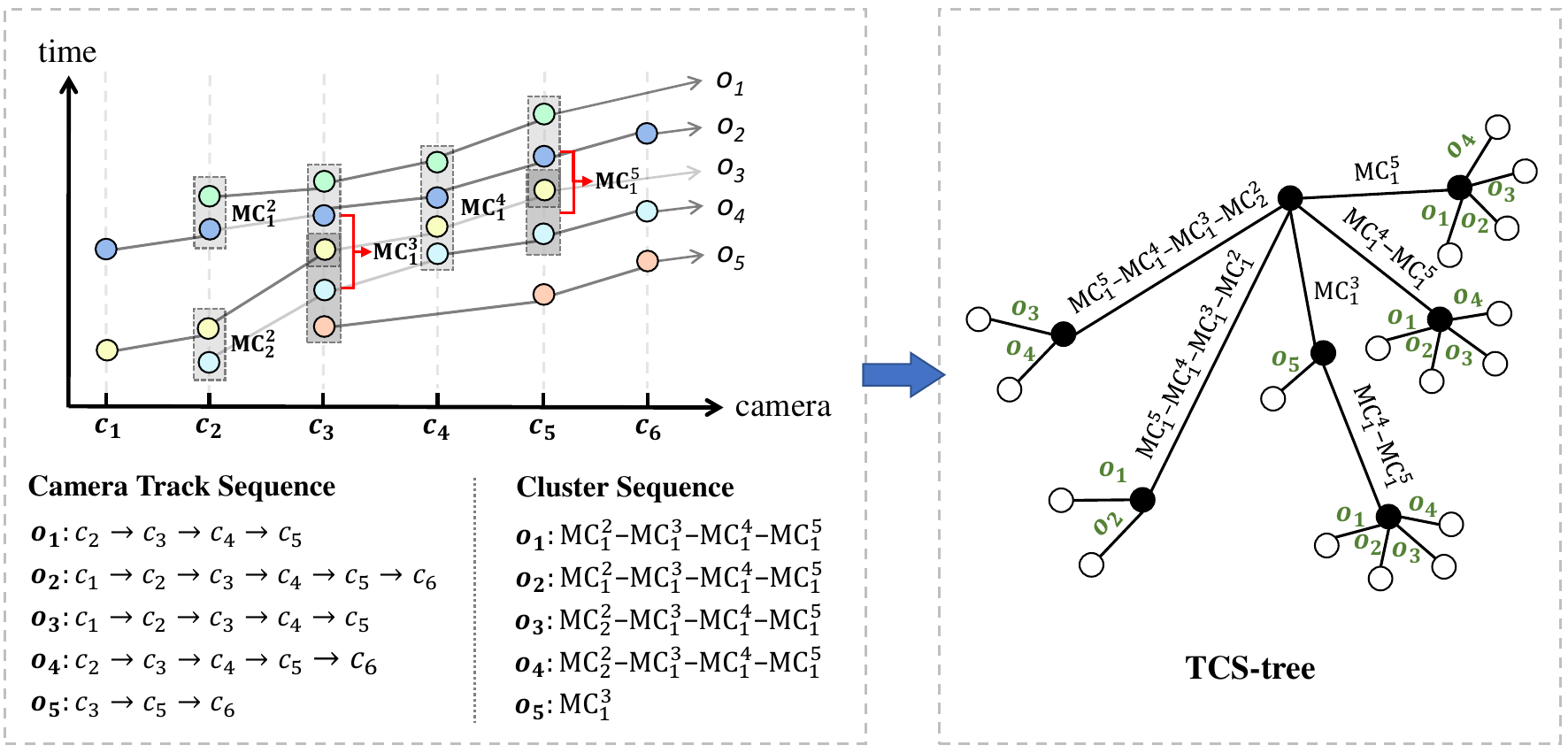}
	\caption{A toy example of the construction of TCS-Tree.}
	\label{fig:index-graph-tcs}
\end{figure}


In the following, we analyze the properties of  $G(o_i)$, which maps an object $o_i$ into a meta-cluster generated by $c_j$. First, the mapping is one-to-one because the temporal dimension is split into disjoint meta-clusters and these meta-clusters cover all objects captured by the camera. In this way, each object can only be assigned to one of the meta-clusters.  Second, for any two objects $o_1$ and $o_2$ that are $\epsilon$-reachable, they will appear in the same meta-cluster. This is also straightforward because the temporal gap between two disjoint meta-clusters is at least $\epsilon$. If $o_1$ and $o_2$ are located at different meta-clusters, they are unable to be $\epsilon$-reachable.



Finally, we can also show that the meta-clusters generate the minimum interval size.

\begin{lemma}
 The meta-clusters minimize $\sum_{j,t} |I_j^t|$.
\end{lemma}

\begin{proof}
We prove this by contradiction. Suppose there exists another mapping function $G'(o_i)$ that generates smaller intervals than the meta-clusters. Then, we can find a pair of objects $(o_i,o_j)$ that belong to the same meta-cluster $MC$, but $G'(o_i)\neq G'(o_j)$.  If $o_i$ and $o_j$ are $\epsilon$-reachable, we finish the proof because the new mapping contradicts with the requirement that $G'(o_i)=G'(o_j)$ if $o_i$ and $o_j$ are $\epsilon$-reachable. If $o_i$ and $o_j$ are not directly $\epsilon$-reachable, based on the construction of meta-clusters, we can find a sequence of objects $S=o_i\rightarrow o_{i1}\rightarrow o_{i2}\rightarrow\ldots o_j$ such that the neighboring objects are $\epsilon$-reachable. Suppose $o_u\rightarrow o_v$ is the first pair  of objects in $S$ satisfying $G'(o_i)=G'(o_u)\neq G'(o_v)$. Then, we find a pair of objects $(o_u,o_v)$ that are $\epsilon$-reachable but $G'(o_u)\neq G'(o_v)$. This also leads to contradiction for mapping function $G'(\cdot)$.
\end{proof}

\subsection{Optimized Verification and Dominance Elimination}
Assisted by TCS-tree, we can perform frequent subsequence mining to quickly identify candidate paths with at least $k$ cameras and traversed by at least $m$ objects. To support efficient verification on these candidate paths, we propose a sliding-window based candidate enumeration strategy and a hashing-based dominance eliminator.

\subsubsection{Sliding-window Based Candidate Enumeration} 
Let $MC_1\rightarrow MC_2\rightarrow \ldots \rightarrow MC_n$ ($n\geq k$) be a candidate sequence of meta-clusters generated by TCS-tree. Each meta-cluster  $MC_i$ is derived from a camera $c_i$ and merged by multiple temporal clusters. An unoptimized verification procedure has to check all combinations of temporal clusters from different cameras, perform object intersection among the temporal clusters, and preserve result sets with at least $m$ objects. The idea of CMC algorithm can be adopted to improve efficiency. We can traverse the cameras and maintain a buffer to store partial intersection results. When we access camera $c_i$, we only need to perform set intersection between the temporal clusters in $c_i$ and the candidate object sets in the buffer. Nonetheless, the procedure requires heavy object set intersection operators. Even though the operation can be optimized to $O(m \, \mathrm{log}(n/m))\, (m \leq n)$~\cite{DBLP:conf/baeza2010fast}, this requires non-trivial implementation skills. In addition, the temporal clusters are overlapped with high redundancy. If we can save cost from unnecessary computation cost incurred by the redundant data, the performance can be further improved.

Algorithm~\ref{alg:sw-verification} depicts our sliding-window based verification algorithm. The sequence of meta-clusters $MC_1\rightarrow MC_2\rightarrow\ldots$ are processed in order. Let $MC_i$ be the current meta-cluster in camera $c_i$. It essentially consists of a sequence of objects ordered by the entrance time to the camera and these objects are grouped into temporal clusters. We access these objects in order and maintain a sliding window $W=[o_l,o_r]$ using the data structure of queue, where $o_r$ is the current object and $o_l$ is the first object that belongs to the same temporal cluster with $o_r$. Then, enqueue and dequeue operations are used for adding or removing an element from the sliding window, respectively.  We also maintain $\mathbb{C}$ to store the candidate patterns generated during the verification process. For each candidate pattern $CP_i=\langle O_i,P_i\rangle$, we construct a queue $Q[CP_i]$ to store the common objects in $O_i$ and the objects in the sliding window $[o_l,o_r]$. The elements in $Q[CP_i]$ will be dynamically updated when the sliding window moves.

When we access the current object $o_r$, we retrieve the set of candidate patterns in $\mathbb{C}$ whose object set contains $o_r$. This step can be efficiently conducted by building an inverted index for the candidate patterns in an online fashion. For each pattern $CP_i$ containing $o_r$, we append $o_r$ to $Q[CP_i]$ with complexity $O(1)$. If $o_r$ is the end of a temporal cluster and $|Q[CP_i]|\geq m$, we find and store the new candidate pattern $\langle Q[CP_i], P_i\rightarrow c_i\rangle$. The sliding window is also updated when $o_r$ switches to a new temporal cluster. We iteratively remove $o_l$ from the sliding window until $o_l$ and $o_r$ belong to the same temporal cluster. When $o_l$ is removed from the sliding window, we also dequeue $o_l$ from $Q[CP_i]$ for those candidate patterns $CP_i$ containing $o_l$.

\begin{example}
We present an example in Figure~\ref{fig:example-sw} to explain the procedure of sliding-window based verification. In this example, we set $m=2$, $k=3$ and the objects with in a  temporal cluster are grouped. Suppose we are now processing meta-cluster $MC_4$ and the candidate patterns after processing $MC_1$, $MC_2$ and $MC_3$ include $CP_1$ and $CP_2$. The objects in $MC_4$ are accessed in the order of $o_3\rightarrow o_4\rightarrow o_2\rightarrow o_5$ and they form two temporal clusters \blue{$\{o_3,o_4,o_2\}$ and $\{o_4,o_2,o_5\}$}. When the first two objects $o_3$ and $o_4$ are accessed, the sliding window is expanded to $[o_3,o_4]$. By checking the inverted index, we also append $o_3$ to $Q[CP_1]$ and $o_4$ to $Q[CP_2]$. When $o_2$ is accessed, it is enqueued into the sliding window, which becomes $[o_4,o_2]$. Since $o_2$ is contained in $CP_2$, we add $o_2$ to $Q[CP_2]$, resulting in a new candidate pattern $CP'=\langle \{o_4,o_2\},c_2\rightarrow c_3\rightarrow c_4\rangle$. Finally, when the last object $o_5$ is accessed, $o_3$ is dequeued from the sliding window and $o_5$ is enqueued. $Q[CP_2]$ is updated and another new candidate pattern $CP''=\langle \{o_4,o_2,o_5\},c_2\rightarrow c_3\rightarrow c_4\rangle$ is generated.

\end{example}

\begin{figure}[h!]
	\centering
		\includegraphics[width=0.45\textwidth]{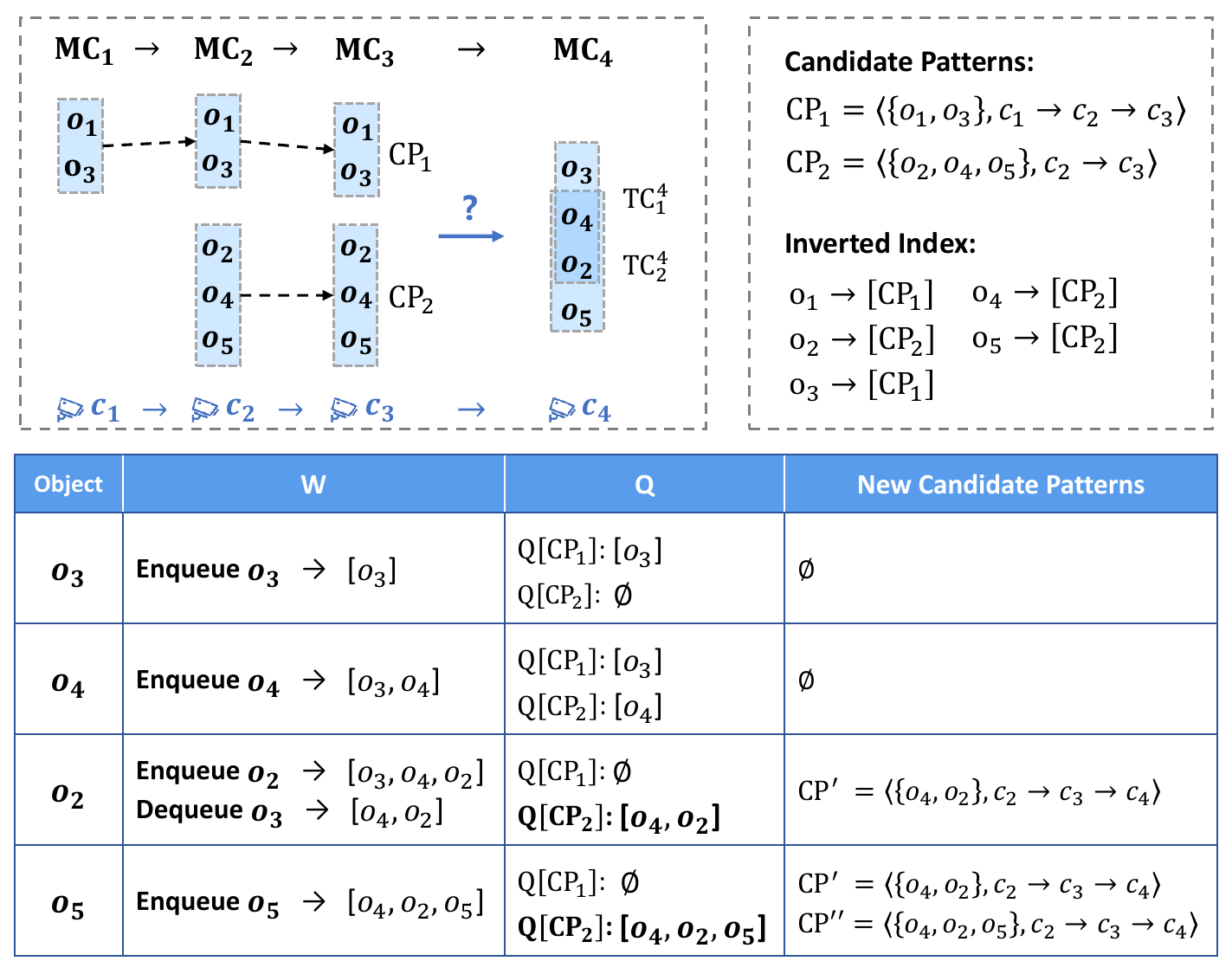}
	\caption{An example of Sliding-window Based Verification}
	\label{fig:example-sw}
\end{figure}

Compared with traditional CMC algorithm that requires frequent set intersection, our sliding-window based verification strategy has two advantages. First, we require no explicit set intersection. The objects are enqueued and dequeued with $O(1)$ complexity to update the common objects shared by the sliding window and candidate patterns. Second, each object is accessed only once, even when it is contained in multiple temporal clusters. This property is useful to handle overlapped temporal clusters with high redundancy and can save considerable computation cost.


\begin{algorithm}[t!]
\SetAlgoNoEnd  \SetAlgoNoLine
\caption{Sliding-window Based Verification}\label{alg:sw-verification}
\small
$\mathbb{C} \leftarrow \emptyset$\;
\For{each $MC_i$ with camera $c_i$ in the frequent sequence}{
    $W \leftarrow$ \textbf{queue}; $Q \leftarrow \emptyset$; $\mathbb{C}_{new} \leftarrow \emptyset$\;
    \For{each candidate pattern $CP_i \in \mathbb{C}$}{
        $Q \left[CP_i\right] \leftarrow$ \textbf{queue}\;
    }
    $S \leftarrow$ objects of $MC_i$ in order of entrance time in $c_i$\;
    \For{each object $o_r$ in $S$}{
        Enqueue $o_r$ into $W$\;
        \For{each $CP_i \in \mathbb{C}$ containing $o_r$}{
            Enqueue $o_r$ into $Q \left[CP_i\right]$\;
        }
        \If{$o_r$ is the end anchor of a temporal cluster}{
            \For{each $CP_i \in \mathbb{C}$}{
                \If{$\left|Q\left[CP_i\right]\right| \geq m$}{
                    $\mathbb{C}_{new} \leftarrow \mathbb{C}_{new} \cup \left \langle Q\left[CP_i\right],P_i\rightarrow c_i\right \rangle$\;
                }
            }
            \While{$o_l$ and $o_r$ in different temporal clusters}{
                Dequeue $o_l$ from $W$ and $Q$\;
            }
        }
    }
    \For{each candidate pattern $CP_i \in \mathbb{C}$}{
        \lIf{$\left| P_i \right| \geq k$}{$\mathbb{R} \leftarrow \mathbb{R} \cup CP_i$}
    }
    \For{each $TC_j \in MC_i$ in the frequent sequence}{
        $\mathbb{C}_{new} \leftarrow \mathbb{C}_{new} \cup \left \langle TC_j, c_i\right \rangle$\;
    }
    $\mathbb{C} \leftarrow \mathbb{C}_{new}$\;
    }
\textbf{return} $\mathbb{R}$\;
\end{algorithm}


    

\subsubsection{Hashing-based Dominance Verification}
In order to identify all the maximal patterns, we need to perform dominance verification among the co-movement patterns in $\mathbb{R}$ and eliminate those that are dominated.  A naive approach is to examine each  $CP_i=\langle O_i,P_i\rangle \in\mathbb{R}$ whether there exists a pattern  $CP_j=\langle O_j,P_j\rangle\in\mathbb{R}$ that dominates it, by verifying whether $O_i\subseteq O_j$ and $P_i\subseteq P_j$. However, this approach is inefficient and entails a considerable amount of unnecessary examination cost when $CP_i$ is not dominated by $CP_j$.

In order to enhance efficiency in dominance verification, we put forward a hashing-based approach. Rather than comparing each $CP_i$ against all other patterns in $\mathbb{R}$, we need only examine whether $CP_i$ is dominated by $CP_j$ with $O_i=O_j$ or $P_i=P_j$, by virtue of the following lemma:

\begin{lemma}
For any pattern $ CP_i=\langle O_i,P_i\rangle$ in the result set $\mathbb{R}$ generated by Algorithm~\ref{alg:sw-verification}, if $CP_i$ is dominated in $\mathbb{R}$, then there must exist another pattern $CP=\langle O,P\rangle \in \mathbb{R}$ such that $O_i=O$ or $P_i=P$.
\end{lemma}
\begin{proof}
We prove by contradiction. Let us assume $CP_i$ is dominated by a pattern $CP_j=\langle O_j,P_j\rangle \in \mathbb{R}$ and we have $O_i\subseteq O_j$ and $P_i\subseteq P_j$. If there is no pattern $CP=\langle O,P\rangle \in \mathbb{R}$  with $O_i=O$ or $P_i=P$, we can conclude that $O_i\neq O_j$ and $P_i\neq P_j$ (i.e., $O_i\subset O_j$ and $P_i\subset P_j$). 

Without loss of generality, we denote $P_i=c_1\rightarrow c_2\rightarrow \ldots\rightarrow c_n$. Since $CP_j=\langle O_j,P_j\rangle$  is a valid pattern, there exists a temporal cluster with object set $O_c\supseteq O_j$ for camera $c_n$. On the other hand, since $O_i$ is a pattern in $\mathbb{R}$, we assume it is generated by the candidate pattern $CP_u=\langle O_u,P_u\rangle$ with $O_u\supseteq O_i$ and $P_u=c_1\rightarrow c_2\rightarrow \ldots\rightarrow c_{n-1}$. Since $O_c$ is also a temporal cluster in camera $c_n$, we know from Algorithm~\ref{alg:sw-verification} that $O_c$ and candidate pattern $CP_u$ can lead to a new pattern $\langle O_c\bigcap O_u,P_i\rangle$, which dominates $CP_i$ because $O_i\subseteq O_c\bigcap O_u$ and their travel path is identical. The pattern $\langle O_c\bigcap O_u,P_i\rangle$ will be added into $\mathbb{R}$ and thus, we find a pattern that leads to a contradiction.
\end{proof}

In this way, we can build two hash maps $HTP$ and $HTO$. $HTP$ uses the travel path as the key and $HTO$ uses the object set as the key. For each candidate pattern $CP_i$, we only need to compare with the patterns in $HTP[P_i]$ and $HTO[O_i]$.

\subsection{Complexity Analysis}

\noindent \textbf{Frequent subsequence mining.} To build TCS-tree, temporal clustering is performed within each camera with linear complexity. Given $N$ objects with average travel path $L$, the complexity is $O(NL)$. Since the merge operation to generate meta-clusters, the construction of suffix tree, and applying suffix tree for frequent subsequence mining are all linear~\cite{ukkonen1995line}, the total complexity in this module is $O(NL)$.

\noindent \textbf{Sliding-window based verification.} Given $M$ frequent subsequences generated from the suffix tree, we need to perform the sliding-window based verification for each sequence of meta-clusters, whose average length is assumed to be $\bar{k}$. From Algorithm~\ref{alg:sw-verification}, we know that the processing cost of each meta-cluster (lines $3$-$21$) is linearly correlated with inverted index size and meta-cluster size. Suppose there are $b$ candidate patterns in the buffer and each contains an average of $\bar{m}$ objects. Let $c$ be the average size of a meta-cluster. Then, the cost of sliding-window based verification is $O(M\bar{k}(b\bar{m}+c))$. 

\noindent\textbf{Hashing-based dominance eliminator}. The complexity of this component is straightforward to estimate. For each candidate pattern $CP_i$, we only need to compare with the patterns in $HTP[P_i]$ and $HTO[O_i]$. So the complexity is determined by the length of $HTP[P_i]$ and $HTO[O_i]$ and the examination cost of subset relationship.

\subsection{Extension to Overlapped Cameras}
In real-world surveillance camera systems, a moving object could be captured by multiple cameras simultaneously, especially in the area of road intersections. To handle the scenario, our idea is to cluster the cameras with overlapped view into a virtual camera V. For any camera $c_i\in V$, there exists $c_j\in V$ such that $c_i$ and $c_j$ have overlapped view. For any two cameras from different clusters, they are not overlapped. Suppose a moving object is captured by multiple cameras within $V$, which we denote as $(c_i,[s_i,e_i]), \ldots, (c_j,[s_j,e_j])$ and these time intervals are overlapped. We can replace them as \blue{$(V, [\min(s_i,\ldots,s_j), \max(e_i,\ldots,e_j)])$}. Hence,  the travel path of each object is represented as a sequence of camera cluster IDs and the associated time intervals. We can still apply the proposed mining algorithms in this paper on the new data model to identify co-movement patterns of grouped objects. 

\eat{
\begin{table*}[t!]
\begin{center}
\caption{Dataset statistics.}
\begin{tabular}{|c|c|c|c|c|c|c|} \hline
& Camera Nodes & Camera Edges & Object Number & Avg Travel Path & Avg staying time & Avg time span for two cameras\\ \hline
Chengdu &  & & & & & \\ \hline
Singapore & 37370 & 540340 & 2756 & 893 & 7 & 71\\ \hline
\end{tabular}
\label{tbl:dataset}
\end{center}
\end{table*}
}

\section{Experiments}\label{sec:exp}

\subsection{Experimental Setup}
\noindent\textbf{Datasets}. Since there lack large-scale and publicly-accessible trajectories recovered from videos, we construct two types of datasets for co-movement pattern mining. In the first group of datasets, we use real GPS trajectories and road network to generate approximate trajectories recovered from videos. We can deploy a specified number of cameras randomly on the road network. The position and view field of each camera are considered as prior knowledge. From the GPS trajectories, we can roughly estimate the entrance time and exit time from a camera according to the travel speed estimation and the attributes of the camera. In this way, we can convert a GPS trajectory  into a sequence of camera IDs and the time interval captured by each camera is also available. In our implementation, we use real trajectory datasets of DIDI Chengdu~\cite{DBLP:journals/pvldb/TongZZCYX18} and Singapore Taxi~\cite{DBLP:journals/pvldb/FanZWT16} as well as their road network to construct two video trajectories for co-movement pattern mining.

In the second group of dataset, we start from raw videos and employ the end-to-end mining pipeline presented in Section~\ref{sec:pipeline} to evaluate the effectiveness of co-movement pattern mining from videos. We adopt Visual Road~\cite{DBLP:conf/sigmod/HaynesMBCC19}, a tool to generate synthetic video benchmark datasets. In this dataset, we generate $1169$ cameras, with $1833$ vehicles moving across the road network.  We apply recent trajectory recovery algorithm~\cite{DBLP:conf/kdd/YuAYZWL22} on these raw videos to generate a sequence of camera id for each vehicle. As the counterpart, we can also derive the GPS trajectories of vehicles from the simulation engine of Visual Road. These GPS trajectories will be used as groundtruth for effectiveness evaluation.

\noindent\textbf{Comparison Approaches}. We compare the proposed \textbf{TCS-tree} with the two baselines derived from previous co-movement pattern miners, namely \textbf{CMC} and \textbf{Apriori}. In addition, we also incorporate the baseline presented in Section~\ref{sec:base-suffix-tree}. The algorithm, denoted by \textbf{FSM}, adopts suffix tree and our proposed sequence-ahead mining framework to mine frequent subsequences. It does not leverage our improved strategies proposed in Section~\ref{sec:tcs-tree}.

\noindent\textbf{Parameter Setup}. We examine the scalability of the proposed algorithms in terms of the six parameters listed in Table~\ref{tbl:query-param}. Among them, $k$, $m$, $\epsilon$ are related to the definition of a co-movement pattern. The remaining parameters are related to dataset settings, including moving object number, camera number and the average travel length per object. We use the same settings for query parameters $k$, $m$ and $\epsilon$ for pattern mining in Chengdu and Singapore datasets. However, the parameters related to dataset settings are different. Overall, the Singapore dataset has \blue{a much} smaller number of distinct objects, but with longer travel path per vehicle. The default parameters are highlighted in bold.

\begin{table}[t!]
\small
\begin{center}
\caption{Parameter setup for scalability analysis.}
\vspace{-3mm}
\begin{tabular}{|c|c|l|} \hline
\multirow{2}{*}{Query path length $k$} & Singapore &  \multirow{2}{*}{$2$, $3$, $4$, $\mathbf{5}$, $6$, $7$, $8$}   \\ \cline{2-2}
& Chengdu &  \\ \hline
\multirow{2}{*}{Group size $m$} & Singapore &  \multirow{2}{*}{$2$, $\mathbf{3}$, $4$, $5$, $6$, $7$, $8$}   \\ \cline{2-2}
& Chengdu & \\ \hline
\multirow{2}{*}{Threshold $\epsilon$} & Singapore &  
 \multirow{2}{*}{$40$, $50$, $\mathbf{60}$, $70$, $80$, $90$, $100$}  \\ \cline{2-2} & Chengdu & \\ \hline
\multirow{2}{*}{Camera number} & Singapore & $10$k, $15$k, $20$k, $25$k, $30$k, $\mathbf{35k}$  \\ \cline{2-3}
& Chengdu & $3$k, $4$k, $5$k, $6$k, $7$k, $8$k, $\mathbf{9k}$ \\ \hline
\multirow{2}{*}{Object number} & Singapore &  $1.5$k, $2.1$k, $2.3$k, $2.5$k, $\mathbf{2.7k}$ \\ \cline{2-3}
& Chengdu & $6$k, $7$k, $8$k, $9$k, $10$k, $11$k, $\mathbf{12k}$ \\ \hline
Average path length & Singapore &  $200$, $300$, $400$, $500$, $600$, $700$, $\mathbf{800}$  \\ \cline{2-3}
($\#$cameras per vehicle) & Chengdu & $90$, $100$, $110$, $120$, $130$, $140$, $\mathbf{150}$ \\ \hline
\end{tabular}
\label{tbl:query-param}
\end{center}
\vspace{-5mm}
\end{table}

All the algorithms are implemented in C++ language and conducted on the server with a 3.20 GHz i9-12900K CPU, 128GB of main memory and 1 TB hard drive. The server runs a Ubuntu Linux release 20.04.5.

\subsection{Scalability Analysis}


\noindent\textbf{Increasing path length $\mathbf{k}$}. In the first experiment, we report the latency of co-movement pattern mining w.r.t. increasing $k$. As shown in Figure~\ref{exp:inc-k-time}, the two baselines are inefficient, with their execution time being at least $100$ times greater than that of our proposed TCS-tree approach. They need to enumerate and verify an enormous number of candidate patterns, necessitating frequent operations of object set intersection. To investigate this phenomenon further, we examine the frequency of object set intersections incurred during the pattern mining process. The findings reveal a direct correlation between this indicator and the overall execution time.  Apriori and CMC trigger  billions of set intersection operations, which become the primary performance bottleneck.  Moreover, we observe that these two baselines are not sensitive to $k$. Since the mining task intends to find non-dominated maximal patterns, they need to enumerate all the candidate paths with varying length.

The FSM algorithm incorporates our proposed sequence-ahead mining framework, surpassing the performance of the baselines and affirming the effectiveness of leveraging suffix trees to filter out false positive candidates. Moreover, as $k$ increases, both FSM and TCS-tree demonstrate enhanced pruning capabilities. The frequent subsequence mining technique expeditiously filters a substantial number of candidate patterns with lengths smaller than $k$. Our TCS-tree method  achieves a significant improvement over FSM. Its index is based on sequences of meta-clusters, which are more fine-grained than the camera sequences in FSM. In addition, it incorporates a sliding-window-based verification and a hashing-based dominance eliminator to further improve efficiency.

\begin{figure}[h!]
  \centering
  \subfigcapskip=-6pt
  \setlength{\abovecaptionskip}{-2pt}
  \subfigure[Chengdu]{
    \includegraphics[width=0.215\textwidth,height=0.14\textheight]{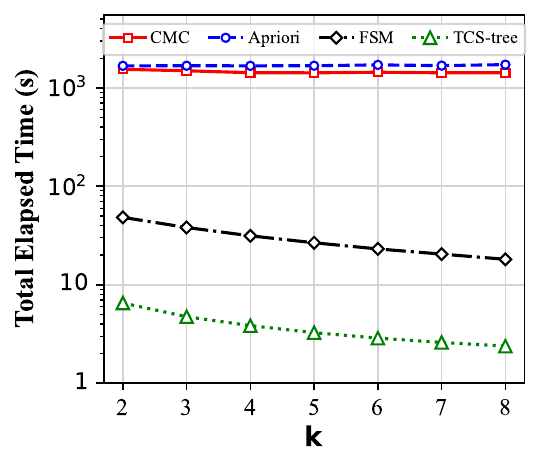}
  }
  \subfigure[Singapore]{
    \includegraphics[width=0.215\textwidth,height=0.14\textheight]{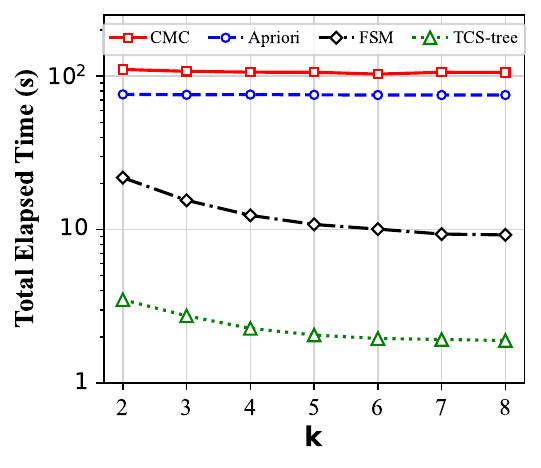}
  }\vskip -6pt

  \subfigure[Chengdu]{
    \includegraphics[width=0.215\textwidth,height=0.14\textheight]{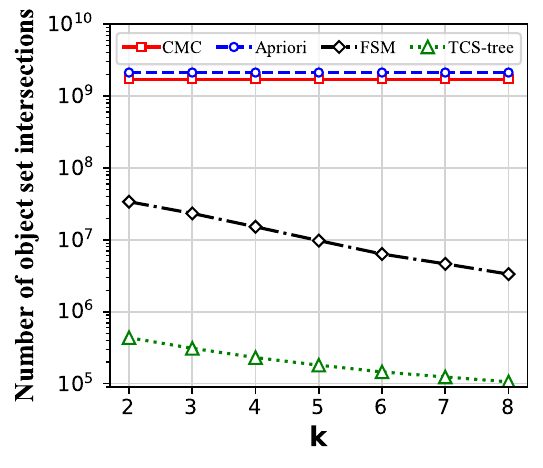}
  }
  \subfigure[Singapore]{
    \includegraphics[width=0.215\textwidth,height=0.14\textheight]{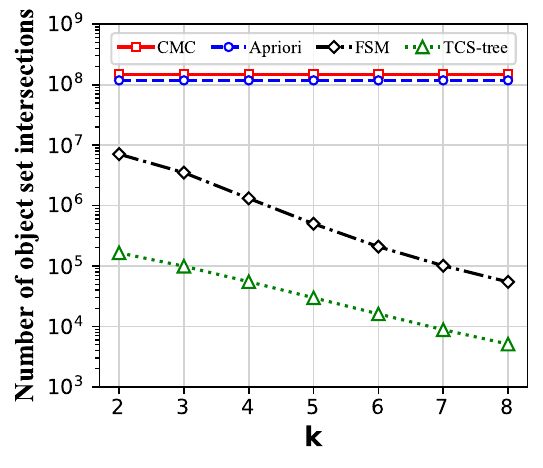}
  }
  \caption{Increasing query path length $k$.}
  \label{exp:inc-k-time}
\end{figure}

\noindent\textbf{Increasing group size $\mathbf{m}$}. The results with increasing group size $m$ are reported in Figure~\ref{exp:inc-m}. It is interesting to find that CMC and Apriori are very sensitive to parameter $m$. When $m$ is large, the number of candidate patterns undergoes a significant reduction, enabling them to even surpass the performance of the FSM algorithm. This behavior can be attributed to the fact that CMC and Apriori retain candidate patterns comprising a minimum of $m$ objects.  When $m$ increases, the number of candidate patterns is sharply reduced, resulting in much lower computation cost for object set intersection. For FSM and TCS-tree, their performances initially benefit from the increase of $m$ because fewer frequent subsequences are generated. When $m$ continues to grow, there exist a small number of candidate patterns and the total computation overhead is dominated by the index construction and frequent subsequence mining.

\begin{figure}[h!]
  \centering
  \subfigcapskip=-6pt
  \setlength{\abovecaptionskip}{-2pt}
  \subfigure[Chengdu]{
    \includegraphics[width=0.215\textwidth,height=0.14\textheight]{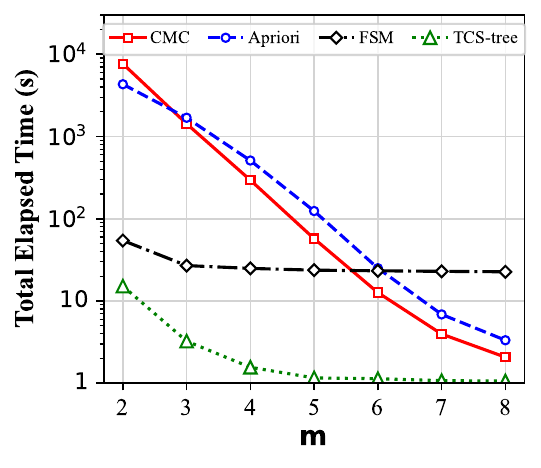}
  }
  \subfigure[Singapore]{
    \includegraphics[width=0.215\textwidth,height=0.14\textheight]{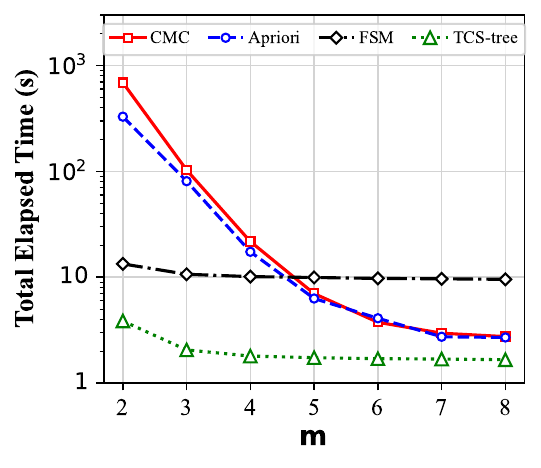}
  }
  \caption{Increasing $\mathbf{m}$.}
  \label{exp:inc-m}
\end{figure}

\noindent\textbf{Increasing proximity threshold $\mathbf{\epsilon}$}. As $\mathbf{\epsilon}$ increases, there are more valid candidate patterns, leading to an escalation in the running time of all the algorithms.  Notably, the FSM algorithm is less sensitive to $\epsilon$ because it does not leverage the parameter for candidate pruning in the stage of frequent subsequence mining. 

\begin{figure}[h!]
  \centering
  \subfigcapskip=-6pt
  \setlength{\abovecaptionskip}{-2pt}
  \subfigure[Chengdu]{
    \includegraphics[width=0.215\textwidth,height=0.14\textheight]{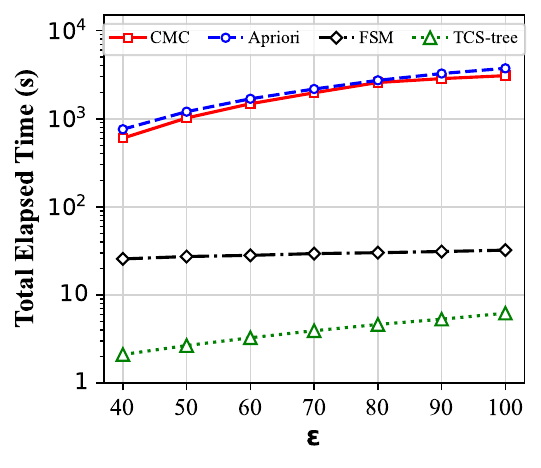}
  }
  \subfigure[Singapore]{
    \includegraphics[width=0.215\textwidth,height=0.14\textheight]{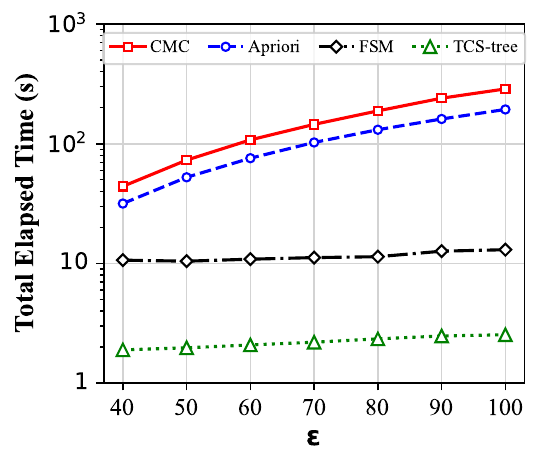}
  }
  \caption{Running time w.r.t. increasing $\mathbf{\epsilon}$.}
  \label{exp:inc-epsilon}
\end{figure}

\noindent\textbf{Increasing dataset-related parameters}. In Figure~\ref{exp:inc-dataset}, we perform scalability analysis w.r.t. concerning the number of cameras, moving objects, and the average length of travel paths for each object. It's not surprising to observe that the figures illustrate an evident ascending trend across all the algorithms.  

Notably, the two baselines display a greater sensitivity to the cardinality of the dataset, including the number of objects and the travel path length of each object. They  exhibit a considerably more pronounced increasing trend in the two factors, when compared with the factor of camera number. Even though CMC applies Apriori algorithm \blue{to} the combination of cameras, the road network ontology remains unchanged. Hence, the deployment of additional cameras in the road network does not yield a significant increase in the number of edges within the derived camera network.


\begin{figure}[h!]
  \centering
  \subfigcapskip=-6pt
  \setlength{\abovecaptionskip}{-2pt}
  \subfigure[Chengdu]{
    \includegraphics[width=0.215\textwidth,height=0.14\textheight]{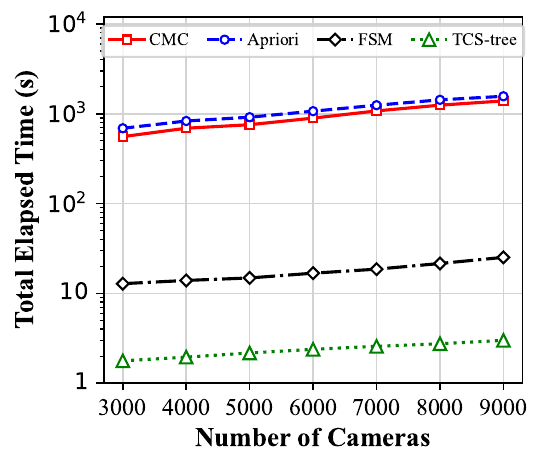}
  }
  \subfigure[Singapore]{
    \includegraphics[width=0.215\textwidth,height=0.14\textheight]{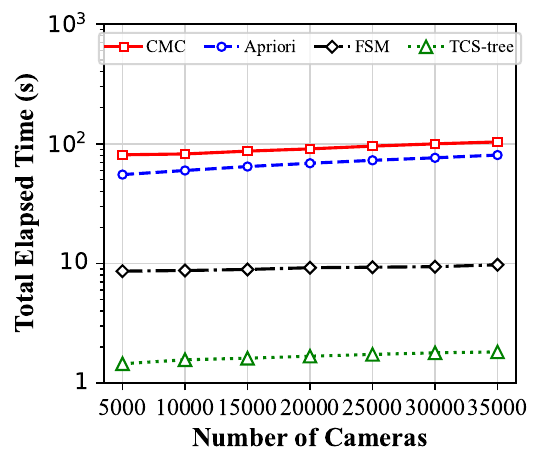}
  }\vskip -6pt
  
\subfigure[Chengdu]{
    \includegraphics[width=0.215\textwidth,height=0.14\textheight]{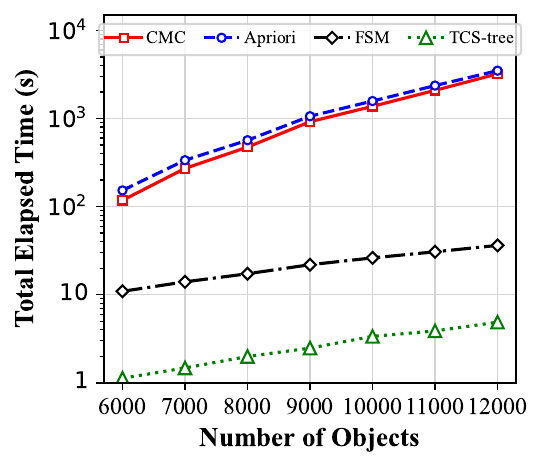}
  }
  \subfigure[Singapore]{
    \includegraphics[width=0.215\textwidth,height=0.14\textheight]{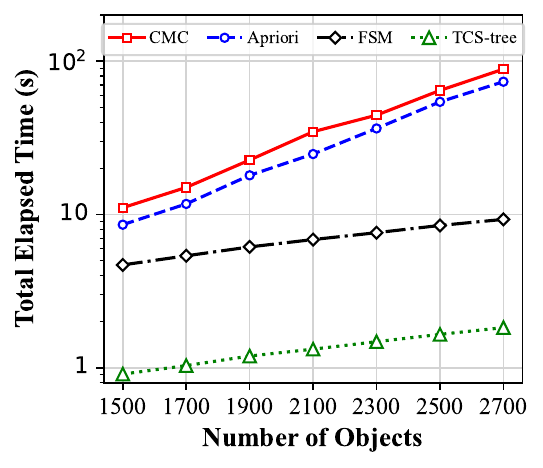}
  }\vskip -6pt
  
   \subfigure[Chengdu]{
    \includegraphics[width=0.215\textwidth,height=0.14\textheight]{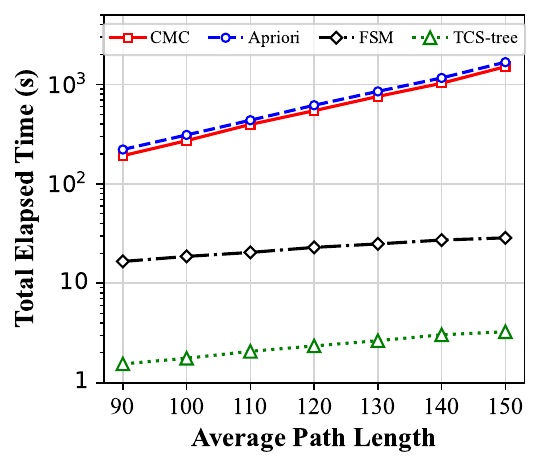}
  }
  \subfigure[Singapore]{
    \includegraphics[width=0.215\textwidth,height=0.14\textheight]{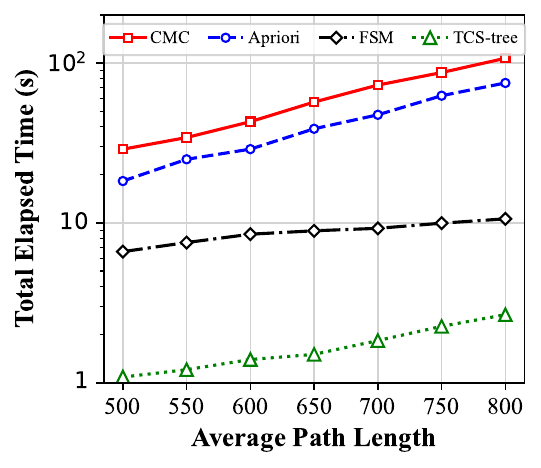}
  }
  \caption{Adjusting dataset-related parameters.}
  \label{exp:inc-dataset}
\end{figure}

\subsection{Break-down \blue{Analysis} and Ablation Study}
This experiment is designed with two goals. Firstly, considering that the TCS-tree algorithm comprises three essential steps, including frequent subsequence mining, verification of candidate patterns, and the elimination of dominated patterns, we conduct a break-down analysis to assess the computational cost associated with each component. Secondly, our technical contributions \blue{originate} from TCS-tree construction, sliding-window-based verification, and hashing-based dominance elimination, \blue{so} we perform an ablation study to evaluate their impact, by devising three variants of our proposed algorithm. The first variant, denoted as TCS-tree-v1, replaces the hashing-based dominance eliminator with the one employed in the baseline methods. The second variant, known as TCS-tree-v2, substitutes the sliding-window-based verification with CMC algorithm. The third variant, namely TCS-tree-v3, builds TCS-tree on top of sequences of camera ids instead of meta-clusters and applies our proposed verification and de-dominance techniques. The goal is to identify the pruning power of fine-grained representation for frequent subsequence mining.  

We run these algorithms with default parameter settings in Table~\ref{tbl:query-param}. From the running time results in Figure~\ref{fig:bk_time},the verification module incurs the highest computational cost, accounting for approximately $75\%$ of the total processing time in the Chengdu dataset. Disabling the hashing-based dominance eliminator in TCS-tree-v1 leads to an almost twofold cost increase in this component. Similarly, removing the sliding-window based verification in TCS-tree-v2 results in a significant cost escalation. The fine-grained data model that represents each object as a sequence of meta-clusters plays the most important role. Substituting this model with sequences of camera IDs dramatically increases the expense of frequent subsequence mining. The reason is that in the original TCS-tree, meta-clusters with size smaller than $m$ can be removed, yielding shorter sequences compared to TCS-tree-v3 and reducing the computational cost for suffix tree construction and frequent subsequence mining. TCS-tree-v3 incurs higher costs for verification and dominance elimination due to its generation of a greater number of candidate patterns. To provide a baseline comparison, we also included FSM in the figure, which lacks the fine-grained data model, sliding-window based verification, and hashing-based dominance eliminator.


\begin{figure}[h!]
  \centering
  \subfigcapskip=-6pt
  \setlength{\abovecaptionskip}{-2pt}
  \subfigure[Chengdu]{
    \includegraphics[width=0.22\textwidth,height=0.14\textheight]{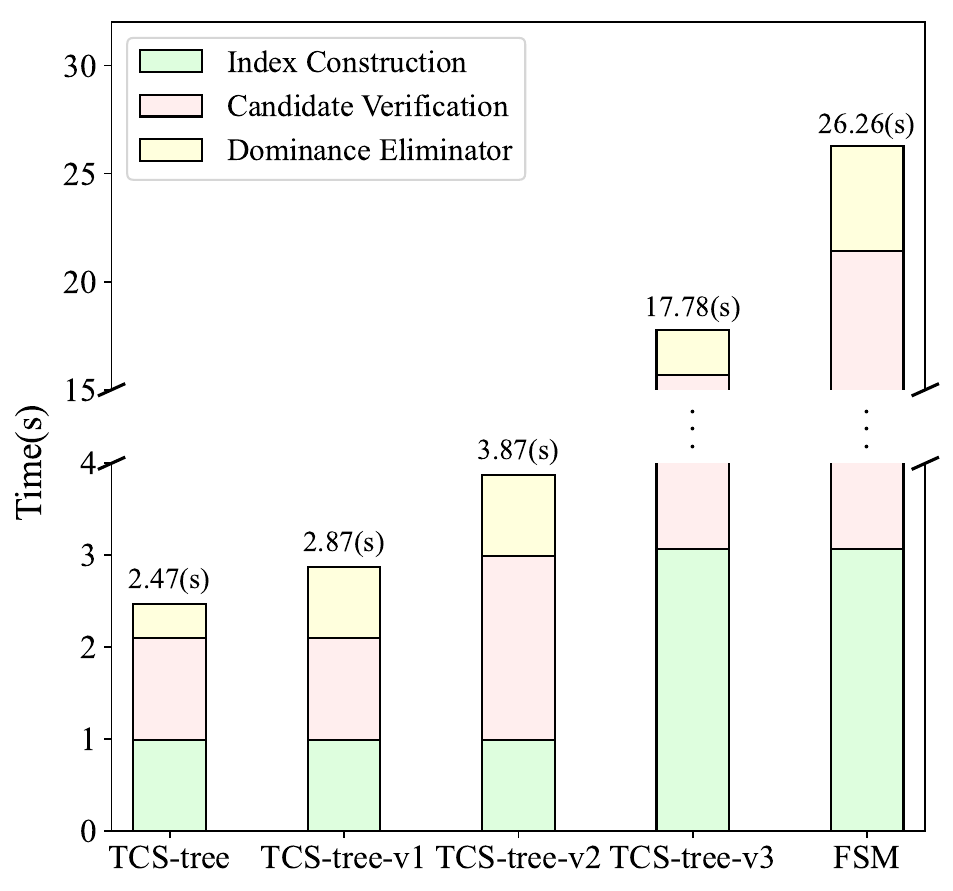}
  }
  \subfigure[Singapore]{
    \includegraphics[width=0.22\textwidth,height=0.14\textheight]{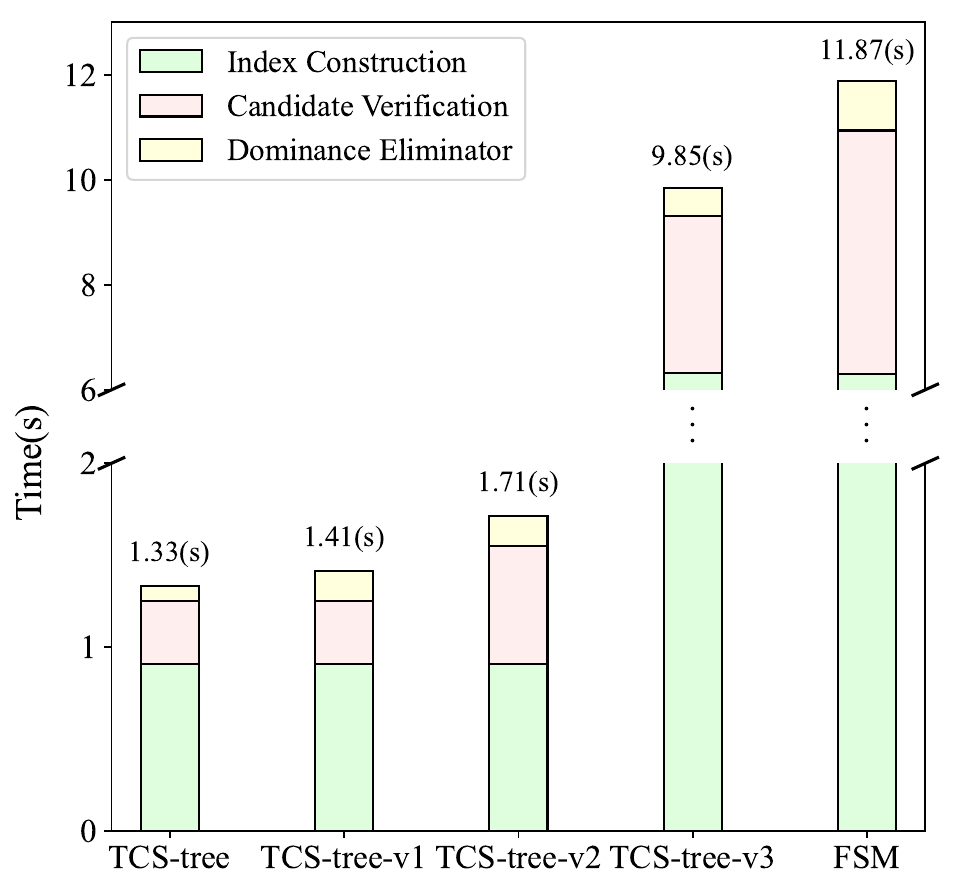}
  }
  \caption{Break-down analysis and ablation study.}
  \label{fig:bk_time}
 \vspace{-4mm}
\end{figure}


\subsection{Effectiveness Analysis}
In the final experiment, our goal is to examine the congruence between the co-movement patterns derived from video recordings and those generated by GPS trajectories. In our datasets generated by Carla, the GPS trajectories can be predetermined by the simulation engine, thereby enabling us to employ them as a reference to generate accurate patterns. Regarding the video data, we utilize the TCS-tree methodology on two distinct data sources. In the first data source, we adopt an existing map matching algorithm to convert the GPS trajectories into camera sequences, which are referred to as oracle travel paths. This is equal to assuming the existence of a perfect trajectory recovery algorithm that generates correct sequences for each moving object. In the second source, we implement the pipeline outlined in Section~\ref{sec:pipeline}, wherein the TCS-tree approach is employed to identify co-movement patterns based on imperfect trajectories extracted from the videos.

We employ the flock pattern mining algorithm ~\cite{DBLP:conf/gis/VieiraBT09} to establish the groundtruth for co-movement patterns derived from GPS trajectories. The flock pattern algorithm necessitates the consideration of spatial proximity between pairs of objects, which aligns with our requirement for $\epsilon$-reachability between object pairs. We utilize the $F_1$-score as the performance metric. A co-movement pattern $CP_1$ generated by our algorithm is deemed to match $CP_2$ from the groundtruth patterns if they both encompass the same set of objects, while ensuring that the intersection over union ($IoU$) between their respective time spans exceeds the threshold of $0.8$.

Regarding the parameters associated with flock patterns, we set  the minimum group size as $2$, the minimum time span as $10$ seconds, and the disk radius as $100$ meters. It takes $36.6$ seconds to mine flock patterns from GPS data, whereas TCS-tree only consumes $0.2$ seconds for co-movement pattern mining from the extracted camera sequences. As to pre-processing step for trajectory recovery from videos, it takes the method proposed in~\cite{DBLP:conf/kdd/YuAYZWL22} $3.5$ hours to perform iterative clustering on the snapshots captured by the cameras. The objects within the same cluster are considered to be the same entity. Fortunately, this pre-processing step is executed only once. Furthermore, in addition to facilitating co-movement pattern analysis, the output of this step can be effectively leveraged to support diverse downstream applications.

As shown in Figure~\ref{fig:F-score}, we report the $F_1$-scores with varying travel path length among the detected patterns. Notably, when executing the TCS-tree algorithm on the oracle camera sequences derived from GPS trajectories, the attained $F_1$-scores  are close to $1$. This outcome suggests a remarkable similarity between co-movement pattern mining driven by video analysis and conventional GPS-driven approaches. Hence, in scenarios where GPS data is inaccessible, a sole reliance on video data sources for co-movement pattern mining emerges as a viable alternative. In practical applications, it is worth noting that trajectory recovery algorithms such as~\cite{DBLP:conf/kdd/YuAYZWL22} inherently generate imperfect travel paths. Consequently, when employing the pipeline introduced in Section~\ref{sec:pipeline} on raw video footage, the  $F_1$-scores exhibit a noticeable decrease, yet they remain within an acceptable range.

\vspace{-3mm}
\begin{figure}[h!]
	\centering
	\includegraphics[width=0.38\textwidth]{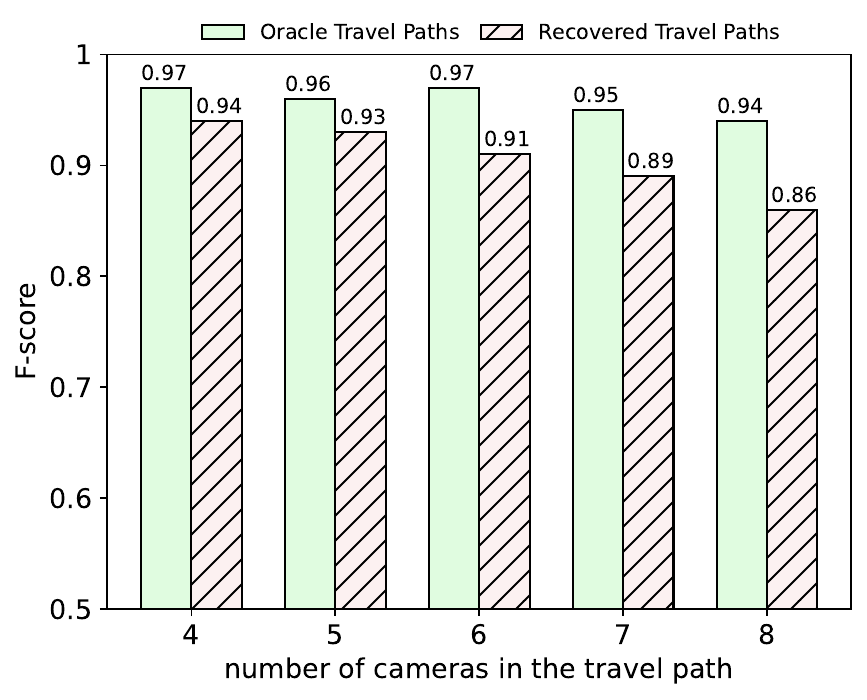}
  \vspace{-3mm}
	\caption{Effectiveness analysis for video-based co-movement pattern mining.}
	\label{fig:F-score}
 \vspace{-5mm}
\end{figure}


\section{Conclusion and Future Study}\label{sec:conclude}

The study at hand makes the first attempt at co-movement pattern mining from video data. We formulate the problem and theoretically prove is NP-hardness. By adopting ideas from existing works, we develop two baselines. We also propose a novel index TCS-tree and an efficient verification strategy based on sliding window and a hashing-based dominance eliminator. Experimental results validate the efficiency of the proposed index and mining algorithm.

While this work is commendable, there remain certain limitations that indicate a wide range of untapped opportunities. Firstly, the focus of the present study is a specific case of co-movement patterns. Future research can aim to explore diversified co-movement pattern definitions. Moreover, it would be worthwhile to consider incorporating temporal gaps as in platoon~\cite{DBLP:journals/dke/LiBK15} into the analysis in subsequent investigations. Secondly, our work builds on trajectories recovered by existing algorithms, which may contain inaccuracies. In light of this, fuzzy pattern mining is an intriguing alternative to extract insights from imperfect trajectories. Thirdly, we note that while co-movement pattern mining has previously been explored in a streaming-based context~\cite{DBLP:journals/pvldb/ChenGFMJG19,DBLP:conf/sigmod/FangGPCMJ20}, the online pattern mining of co-movement patterns from video data remains largely unexplored, making it a promising avenue for future research. Finally, our study extracts discrete trajectories from video data. It would be interesting to explore co-movement pattern mining using alternative sensory data with approximate or discrete trajectories. For example, bluetooth-based contact tracing devices  infer user proximity via Bluetooth signal strength. In this scenario, we are aware of the proximity relationship between user pairs, without knowing their explicit positions. This results in a completely different data model and requires new problem formulation and mining algorithms. Additionally, bluetooth data enables indoor positioning based on proximity to mounted beacons, while cellular data approximates positions relative to nearby base stations.  We can explore the extension of existing GPS-based co-movement pattern mining methods to these two scenarios with new proximity constraint on the approximate position information.  
\bibliographystyle{ACM-Reference-Format}
\bibliography{sample-base}


\end{document}